\documentclass[runningheads]{llncs}
\usepackage{citesort}
\usepackage{graphicx}
\usepackage{theorem}
\usepackage{latexsym}
\usepackage{multirow}
\usepackage{algorithm,algorithmic}
\usepackage{amssymb} 
\usepackage{xcolor}
\usepackage{multirow,eepic}
\usepackage{wrapfig}
\usepackage{lipsum}

\makeatletter
\def\senbun#1(#2)#3({\@senbun(#2)(}
\def\@senbun(#1,#2)(#3,#4){%
   \@tempdima#1\p@ \advance\@tempdima#3\p@
   \divide\@tempdima\tw@
   \@tempdimb#2\p@ \advance\@tempdimb#4\p@
   \divide\@tempdimb\tw@
   \edef\@senbun@temp{\noexpand\qbezier(#1,#2)%
      (\strip@pt\@tempdima,\strip@pt\@tempdimb)(#3,#4)}%
   \@senbun@temp}
\makeatother

\newcommand{\BbbC}{{\rm\kern.22em\rule[.1ex]{.06em}{1.4ex}\kern-.28em C}} 
\newcommand{\BbbQ}{{\rm\kern.22em\rule[.1ex]{.06em}{1.4ex}\kern-.28em Q}}

\newcounter{Codeline}
\newcommand{\Newcodeline}{\setcounter{Codeline}{1}}
\newcommand{\Cl}{{\theCodeline}: \addtocounter{Codeline}{1}}
\newcommand{\crm}{\\}

\newcommand{\N}{{\rm I\kern-.22em N}} 
\newcommand{\Z}{{\sf Z\kern-.42em Z}} 
\newcommand{\R}{{\rm I\kern-.22em R}} 

\newcommand{\LU}{{\mathcal{LUMI}}} 
\newcommand{\LUM}{{\mathcal{LUMI}}}

\newcommand{\OB}{{\mathcal{OBLOT}}} 
\newcommand{\LK}{{\mathit{Look}}}

\newcommand{\M}{{\mathit{Move}}}

\newcommand{\conf}{{\cal C}}

\newcommand{\aunfair}{A_{\mathit{unfair}}}
\newcommand{\look}{{\mathit{Look}}}
\newcommand{\comp}{{\mathit{Compute}}}
\newcommand{\move}{{\mathit{Move}}}

\authorrunning{R. Nakai et.al.}
\titlerunning{Gathering Problems for Autonomous Mobile Robots with Lights in ASYNC}

\begin{document}

\title{Asynchronous Gathering Algorithms for Autonomous Mobile Robots with Lights}
\author{Rikuo~NAKAI\inst{1} \and Yuichi~SUDO\inst{2} \and Koichi~WADA\inst{3}}
\institute{
  Graduate School of Science and Engineering, Hosei University,
  \and
  Faculty of Computer and Information Sciences, Hosei University,
  \and
  Faculty of Science and Engineering, Hosei University
}
\maketitle

\begin{abstract}
  We consider  a {\em Gathering} problem for $n$ autonomous mobile robots with persistent memory called {\em  light}
  in an asynchronous scheduler (ASYNC).
  It is well known that Gathering is impossible when robots have no lights in basic common models,
  if the system is semi-synchronous (SSYNC) or even centralized (only one robot is active in each time).
  It is known that Gathering can be solved by robots with $10$ colors of lights in ASYNC. This result is obtained by combining the following results. (1) The simulation of SSYNC robots with $k$ colors by ASYNC robots with $5k$ colors~\cite{DFPSY}, and (2) Gathering is solved by SSYNC robots with $2$ colors~\cite{TWK}.
  
  In this paper, we improve the result by reducing the number of colors and show that Gathering can be solved by ASYNC robots with $3$ colors of lights.
  We also show that we can construct a simulation algorithm of any \emph{unfair} SSYNC algorithm using $k$ colors by ASYNC robots with $3k$ colors, where unfairness does not guarantee that every robot is activated infinitely often. Combining this simulation and the Gathering algorithm by SSYNC robots with $2$ colors~\cite{TWK}, we obtain a Gathering algorithm by ASYNC robots with $6$ colors. Our main result can be obtained by reducing the number of colors from $6$ to $3$.
\end{abstract}

\section{Introduction}\label{sec:intro}

\subsection{Background and Motivation}

The computational issues of autonomous mobile entities %
have been the object of much research in  distributed computing. In this paper, we focus on mobile objects operating on a two-dimensional Euclidean space but there are several research on three-dimensional spaces and graphs~\cite{DBLP:series/lncs/11340}. Each robot operate in $\look$-$\comp$-$\move$ ($\mathit{LCM}$) cycles.
 In the $\mathit{Look}$ phase, an entity, viewed as a point and usually called {\em robot},  obtains a snapshot of the space; in  the $\mathit{Compute}$ phase
it  executes its algorithm (the same for all robots) using the snapshot as input; it then moves towards the computed destination in the $\mathit{Move}$ phase.
Repeating these cycles, the robots are able to collectively perform some tasks and solve some problems. 
The research interest has been on determining  the impact that
{\em internal}  capabilities (e.g., memory, communication) and {\em external}
 conditions (e.g., synchrony, activation scheduler) have on the solvability of a problem.

 We also explore such weakest capabilities to solve the task. The problem considers in this paper is {\em Gathering}, which is one of the most fundamental tasks of autonomous mobile robots. Gathering is the process of $n$ mobile robots, initially located on arbitrary positions, meeting within finite time at a location, not known a priori. When there are only two robots, this task is called {\em Rendezvous}.
Since Gathering and Rendezvous are simple but essential problems, they have been intensively studied,
and a number of possibility and/or impossibility results have been shown under the different assumptions\cite{AP,AOSY,BDT,CFPS,DGCMR,DKLMPW,DP,IKIW,ISKIDWY,KLOT,KLASING200827,LMA,P,SDY}.
The solvability of Gathering and Rendezvous depends on the activation schedule and the synchronization level.
Usually three basic types of schedulers are identified, the fully synchronous (FSYNC), the semi-synchronous (SSYNC) and the asynchronous (ASYNC)\footnote{In addition to these basic models , the new model semi-asynchronous (SAsync)~\cite{DBLP:journals/access/CiceroneSN21} is recently proposed to reveal the gap between SSYNC and ASYNC.}.
Gathering and Rendezvous are trivially solvable in FSYNC and the basic model.
However, these problems cannot be solved in SSYNC without any additional assumptions \cite{FPS}, and the same is true in ASYNC.
In particular, Gathering is not solvable even  in a restricted subclass of SSYNC scheduler,
where exactly one robot is activated in each round and always in the same order (called ROUND-ROBIN) \cite{DGCMR}. %
If all robots are initially located on different positions (called distinct Gathering),
this version of the problem is not solvable even in the CENT scheduler,
in which exactly one robot is activated in each round~\cite{DGCMR}.%

In \cite{DFPSY}, persistent memory called {\em light} has been introduced to reveal relationship between ASYNC and SSYNC and they show asynchronous robots with lights equipped with a constant number of colors, are strictly more powerful than semi-synchronous robots without lights: %
for any algorithm $\mathcal{A}$ designed for semi-synchronous robots (without colors), 
they give a simulation algorithm by which
asynchronous robots with $5$ colors simulate an execution of $\mathcal{A}$.
Rendezvous can be solved by robots with lights without any other additional assumptions~\cite{FSVY,DFPSY,V}.
Gathering is also solvable by robots with lights and it can be solved by robots with $2$ colors of lights in SSYNC~\cite{TWK}.
The power of lights to solve other problems are discussed in \cite{DBLP:journals/iandc/DEmidioSFN18,LFCPSV,DFN}. %

\subsection{Our Contribution}

In this paper, we study Gathering algorithms by robots with lights in the most realistic  schedulers, ASYNC and some of the weakest conditions in term of computational power. As for Gathering algorithms in ASYNC, the following results are known;
Cielieback et al.~\cite{CFPS} solves the distinct Gathering for more than two robots with weak multiplicity detection, where the distinct gathering means all robots are initially placed in different positions, and weak multiplicity detection helps a robot to identify multiple occurrences of robots at a single point. Bhagat et al.~\cite{BM} solves a gathering problem for five or more robots under the additional constraint to minimize the maximum distance traversed by any robots, with weak multiplicity detection or $4$ colors of lights. Both algorithms work without any extra assumptions like agreements of coordinate systems, unit distance and chirality and rigidity of movement, but they solve some constrained gathering problem. 
Gathering can be solved by ASYNC robots with $10$ colors of lights by using the following two results;
\begin{enumerate}
    \item[(a)] Any algorithm in SSYNC using $k$ colors of light can be simulated by ASYNC robots with $5k$ colors of lights~\cite{DFPSY}, and 
    \item[(b)] A Gathering algorithm is constructed by SSYNC robots with $2$ colors of lights~\cite{TWK}. 
\end{enumerate}

Since the algorithm shown in~\cite{TWK} needs the chirality assumption but no any other extra assumptions, the obtained algorithm works only with chirality.

This paper improves the result just stated above by reducing the number of colors and shows that Gathering can be solved by ASYNC robots with $3$ colors of lights in no any other extra assumptions except chirality as follows;

\begin{enumerate}
    \item We construct a simulation algorithm of any \emph{unfair} SSYNC algorithm using $k$ colors by ASYNC robots with $3k$ colors of lights, where unfair SSYNC is that the adversary makes enabled robots (changing its color or moving a different location) active in SSYNC.
We have reduced the number of colors used in the simulation to $3k$ from $5k$, although the simulated algorithms are limited to ones working in unfair SSYNC.
Since many robots algorithms seem to work in unfair SSYNC if it works in (fair) SSYNC, this simulation is interesting in itself and can be used to reduce the number of colors used in algorithms working in ASYNC.

\item We show that the Gathering algorithm with $2$ colors of light shown in~\cite{TWK} can still work in unfair SSYNC. Hence we obtained that Gathering can be solved by ASYNC robots with 6 colors of lights in no any other extra assumptions except chirality.
The Gathering algorithm of~\cite{TWK} is divided into two sub-algorithms. The first one makes a configuration that all robots are located on one straight line from any initial configuration and the second one is a Gathering algorithm from any initial configuration such that all robots are located on the straight line.
The first one needs no lights but assumption of chirality and
the second one need no extra assumptions and uses $2$ colors. 
We show that the both algorithms can work in unfair SSYNC by defining a potential function for each algorithm and showing that each function becomes monotonically decreasing for any behaviour of each algorithm.  

\item We improve the number of colors used in the algorithm into $3$. Since the second algorithm uses $2$ colors in SSYNC, the resultant algorithm have $6$ colors. Thus in order to reduce the number of colors, we directly construct a $3$-color Gathering algorithm from any configuration such that all robots are located on the straight line in ASYNC. Combining with the simulation of the first algorithm,
we have obtained an ASYNC Gathering algorithm with $3$ colors of light.
\end{enumerate}

We have used the method by combining the simulation of SSYNC robots by ASYNC robots and algorithms working in SSYNC to reduce the number of colors used in the resultant algorithm. 
Proving the correctness of algorithms working in ASYNC is complicated and error-prone.
However, our combination method can reduce the complexity of proving the correctness since the proof of working in  ASYNC is divided into the proof of correctness of the simulation and the simulated algorithm working in SSYNC, both of which are easier than that in ASYNC.

\section{Model and Preliminaries}%

\subsection{The Basics}

The systems considered in this paper consist of a team  $R = \{ r_0 ,\cdots,
r_{n-1}\}$ of  computational entities moving and operating
 in the Euclidean plane $\mathbb R^2$. Viewed as  points, and called {\em robots},
the entities can move freely and continuously in the plane.
Each robot has its own local coordinate system and it always perceives itself at its origin;
there might not be consistency between these coordinate systems.
A robot is equipped with sensorial devices that allows it to   observe the positions of the other robots in its local coordinate system.

The robots are {\em identical}: they are indistinguishable by their appearance and they execute the same protocol. The robots are {\em autonomous}, without a central control.  

At any point in time, a robot is either {\em active} or {\em inactive}. Upon becoming active, a robot $r_i$ executes a $\mathit{ Look}$-$\mathit{Compute}$-$\mathit{Move}$ ($\mathit{LCM}$) cycle performing the following three operations:
\begin{enumerate}
\item {$\look$:} The robot activates its sensors to obtain a snapshot of the positions occupied by robots with respect to its own coordinate
system\footnote{This is called the {\em full visibility} (or unlimited visibility)  setting; restricted forms of visibility have also been considered for these systems.}. The snapshot of $r_i$ is denoted as ${\cal SS}_i$.
\item {$\comp$:} The robot executes its algorithm using the snapshot as input. The result of the computation is a destination point.
\item {$\move$:} The robot moves in a straight line toward  the computed destination but the robot may be stopped by an adversary before reaching the computed destination. 
In this case, the movement is called {\em non-rigid}. Otherwise, it is called {\em rigid}.
When stopped before reaching its destination in the non-rigid movement, a robot moves at least a minimum distance $\delta >0$. If the distance to the destination is at most $\delta$, the robot can reach it. 
We assume non-rigid movement throughout the paper.
If the destination is the current location, the robot stays still.
 \end{enumerate}
 
\noindent When inactive, a robot is idle. All robots are initially idle. The amount of time to complete a cycle is assumed to be finite, and the $\mathit{Look}$ operation is assumed to be instantaneous.

There might not be consistency between the local coordinate systems and their unit of distance. 
The absence of any a-priori assumption on consistency of the local coordinate systems is called {\em disorientation}.

The robots are said
to have {\em chirality} if they share the same circular orientation of the plane (i.e., they 
agree on ``clockwise'' direction). If there is chirality, then there exists a unique circular ordering of locations occupied robots~\cite{SY}.
Thus, for each edge of the convex hull obtained by locations of $n$ robots ($n \geq 3$), all robots can agree with the right vertex of the edge.
\subsection{The Models}\label{sec:model}
Different models,  based on the same basic premises defined above, have been considered in the literature
and we will use the following two models.

In the most common model, $\OB$, the robots are {\em silent}: they have no explicit means of communication; furthermore they are {\em oblivious}: at the start of a cycle, a robot has no
memory of observations and computations performed in previous cycles.

In the other common model, $\LUM$, 
each robot $r_i$ is  equipped with a persistent  visible
state variable $\ell_i$, called {\em light}, whose values are taken from a finite set $C$ of states called {\em colors} (including the color that represents the initial state when the light is off). 
The colors of the lights can be set in each cycle by $r_i$ at the end of its {\em Compute} operation. 
A light is {\em persistent} from one computational cycle to the next: the color is not automatically reset at the end of a cycle;  the robot is otherwise oblivious, forgetting all other information from previous cycles.
In $\LUM$, the {\em Look} operation produces a colored snapshot; i.e., it returns the set of pairs 
 $(position,color)$ 
of the other robots\footnote{If  (strong) multiplicity detection is assumed, the snapshot is a multi-set.}.
Note that if $|C|=1$, then the light is not used; thus, this case corresponds to the $\OB$ model. 

We denote by $\ell_i(t)$ the color of light $r_i$ has at time $t$ and $p_i(t) \in \R^2$ the position occupied by robot $r_i$ at time $t$ represented in some global coordinate system. A {\em configuration} ${\cal C}(t)$ at time $t$ is a multi-set of $n$ pairs $(\ell_i(t),p_i(t))$, each defining the color of light and the position of robot $r_i$ at time $t$.
When no confusion arises, ${\cal C}(t)$ is simply denoted by $\cal C$.

If a configuration is that robots are located on a line segment connecting $p$ and $q$ (denoted as $pq)$, 
this configuration is denoted by a regular-expression-like sequence of colors robots have from the endpoint $p$ to the other endpoint $q$. 
Formally we define {\em color-configurations} for a configuration of line segment $pq$ as follows;
Let $\conf(t)$ be the line-segment configuration at time $t$.%
{\em Color-configurations} for $\conf(t)$ are defined as (0)-(3) as follows;
\begin{enumerate}
\item[(0)] Factor $f$ is defined as either $\alpha$, $(\alpha|\beta)$ or $(\alpha|\beta|\gamma)$, where $\alpha$, $\beta$ and $\gamma$ are colors and $(\alpha|\beta)$  and $(\alpha|\beta|\gamma)$) denote $\alpha$ or $\beta$ and $\alpha$,  $\beta$ or $\gamma$, respectively. Color(s) which robots at a point have are denoted as $f$. Let $f$, $g$ and $h$ be factors.
\item[(1)] $fg$ denotes a configuration that 
all robots at $p$ have colors $f$, all robots at $q$ have colors $g$, and
there are no robots inside  the segment.
\item[(2)] $f g^+ h$ denotes a configuration that
all robots at $p$ have colors $f$, all robots at $q$ have colors $h$, and %
there exists at least one point inside the segment where all robots located there have colors $g$.%
\item[(3)] $f g_m h$, 
if all robots at $p$ have colors $f$, all robots at $q$ have colors $h$, and 
all robots at the mid-point of the segment have colors $g$ and there are no robots except on the three locations.%
\end{enumerate}

If a color-configuration is one of  (1) $\sim$ (3), it is denoted as $f g^* h$.
Let $dis(\conf(t))$ denote the length of the segment in the configuration $\conf(t)$.%
The color-configuration for $\conf(t)$ is denoted as $cc(\conf(t))$ and the number of points which have  color $\alpha$ in $\conf(t)$ is denoted as $\#_{\alpha}(\conf(t))$.
We also use this notation for a snapshot of robot.

In Section~\ref{sec:GatheringAlgorithm}, color-configurations defined here are used, and in addition, one abuse of notation is used as follows. Letting $f$, $g$, and $h$ be factors, $f^+gh^*$ denotes that all robots at $p$ have colors $f$, all robots at $q$ have colors $g$ or $h$, and all robots inside the segment have colors $f$, $g$, and $h$.%

\subsection{The Schedulers}

With respect to the activation schedule of the robots, and the duration of their LCM cycles, the fundamental distinction is between the \emph{asynchronous} and \emph{synchronous} settings.

In the \emph{synchronous} setting (SSYNC), also called semi-synchronous, time is divided into discrete intervals, called \emph{rounds}; in each round some robots are activated simultaneously, and perform their LCM cycle in perfect synchronization.

A popular synchronous setting which plays an important role is the 
\emph{fully-synchronous} setting (FSYNC), where every robot is activated in every round; that is, the activation scheduler has no adversarial power.

In the \emph{asynchronous} setting (ASYNC), there is no common notion of time, each robot is activated independently of the others, the duration of each phase is finite but unpredictable and might be different in different cycles.
In this paper, we are concerned with ASYNC and we assume the following;
In a $\mathit{Look}$ operation, a snapshot of the environment is taken at some time $t_L$  and we say that the \emph{$\mathit{Look}$ operation is performed at time $t_L$.}
Each $\mathit{Compute}$ operation of $r_i$ is assumed to be done at time $t_C$ and the color of its light $\ell_i(t)$ and its pending destination $des_i$ are both set to the computed values for any time greater than $t_C$
\footnote{Note that if some robot performs a $\mathit{Look}$ operation at time $t_C$, then
it observes the former color and if it does at time $t_C+\epsilon (\forall\epsilon>0)$, then
it observes the newly computed color.}.
When the movement in a $\mathit{Move}$ operation begins at time $t_B$ and ends at $t_E$, we say that it is performed during interval $[t_B, t_E]$, and the beginning (resp. ending) of the movement is denoted by $\mathit{Move_{BEGIN}}$ (resp. $\mathit{Move_{END}}$) occurring at time $t_B$ (resp. $t_E$).  
In the following, $\mathit{Compute}$, $\mathit{Move_{BEGIN}}$ and $\mathit{Move_{END}}$ are abbreviated as $\mathit{Comp}$, $\mathit{M_{B}}$ and $\mathit{M_{E}}$, respectively.
When a cycle has no actual movement (i.e., robots only change color and their destinations are the current positions),
we can equivalently assume that the $\mathit{Move}$ operation in this cycle is omitted, since we can consider 
the $\mathit{Move}$ operation to be performed just before the next $\mathit{Look}$ operation.

Without loss of generality, we assume the set of time instants at which the robots start executions of $\mathit{Look}$, $\mathit{Comp}$, $\mathit{M_B}$ and $\mathit{M_E}$ to be $\N$. We also assume the followings for each operation.
\begin{enumerate}
\item $\mathit{Comp}$ operation is performed instantaneously at integer time $t_C$ and if some robot performs a $\LK$ operation at time $t_C$, then
it observes the former color and if it does at time $t_C+1$, then
it observes the newly computed color.
\item When the movement in a $\M$ operation begins at $t_B$ and ends at $t_E \geq t_B+1$ and if a robot performs a $\LK$ operation at time $t_B$ then it observes the location before moving and it does at time $t (t_B+1\leq t \leq t_E)$, then it observes any location on the half-open line segment between one before moving (inclusive) and the destination (exclusive) satisfying the following condition, letting $p_{t}$ be the location of the moving robot at time $t$, for times $t$ and $t'$ such that $t_B+1\leq t<t' \leq t_E$, it holds that   $dis(p_{t_B},p_{t}) <dis(p_{t_B},p_{t'})$, where $dis(p,q)$ denotes the distance between $p$ and $q$. The selected location is assumed to be determined by adversary.  Also if it does at time $t_E+1$, it observes the destination.
\end{enumerate}

In SSYNC and ASYNC settings, the selection of which robots are activated is made by an adversarial scheduler, whose only limit is that every robot must be activated infinitely often (i.e., it is a \emph{fair} scheduler). We also consider an \emph{unfair} scheduler. When a robot becomes active and performs the LCM cycle, the robot is \emph{enabled} if it changes its color and/or the computed  destination is different from the current position.
The unfair scheduler does not guarantee that every robot is activated infinitely often.
It is only guaranteed that if there is one or more enabled-robots at a time $t$, at least one enabled-robot will be activated or become non-enabled at some time $t' > t$. 
Note that in a computation under this scheduler, an enabled-robot may not be activated until it becomes the only enabled robot.

\section{Simulating Algorithms in unfair SSYNC by ASYNC $\LU$ robots}
\label{sec:SIM}
In this section, we show that any $\OB$ algorithm working in unfair SSYNC can be simulated by $\LU$ robots with $3$ colors in ASYNC.

\subsection{Simulation in ASYNC for Algorithms in unfair SSYNC}

\Newcodeline
\begin{algorithm}[ht]
  \caption{SIM-for-Unfair($r_i$)[$A_{\mathit{unfair}}$]}
  \label{algo:SIM}
  {\footnotesize
    \begin{tabbing}
      111 \= 11 \= 11 \= 11 \= 11 \= 11 \= 11
      \= \kill
      {\em Assumptions}: non-rigid, $\LU$, $\ell_i$ has $3$ colors($S$, $M$, and $E$), 
      initially $\ell_i=S$;\crm
      {\em Input}: $A_{\mathit{unfair}}$ : algorithm working in unfair SSYNC, snapshot ${\cal SS}_i$ of $r_i$;\crm

      \Cl \> {\bf case}  $cc({\cal SS}_i)$  {\bf of } \crm

      \Cl \> $\in \forall S$: \crm
      \Cl \> \> {\bf if} $r_i$ is  $\aunfair$-enabled {\bf then}\crm
      \Cl\label{ch-alg-here} \> \> \> $r_i$ executes $\aunfair$\crm
      \Cl \> \> \> $des_i \leftarrow$ the computed destination of $\aunfair$\crm
      \Cl \> \> \> $l_i \leftarrow M$\crm
      \Cl \> \> {\bf else} do nothing\crm
      \Cl \> $\in \forall S,M$: \crm
      \Cl \>\> $l_i \leftarrow M$\crm
      \Cl \> $\in \forall M$ {\bf or} $\forall M,E$: 
      \crm
      \Cl \> \> $l_i \leftarrow E$ \crm
      \Cl \> $\in \forall E$ {\bf or} $\forall S,E$:
      \crm
      \Cl \> \> $l_i \leftarrow S$ \crm
      \Cl \> {\bf endcase}
    \end{tabbing}
  }
\end{algorithm}

We show an algorithm in ASYNC that simulates algorithms in unfair SSYNC. The algorithm is shown in Algorithm~\ref{algo:SIM}. The algorithm in square brackets indicates that it is given as input of a robot and simulated. Its transitions of colors of light is shown in Figure~\ref{fig:sim}.The algorithm in square brackets indicates that it is given as input of a robot $r_i$ and simulated. Let $\aunfair$ be a simulated algorithm in unfair SSYNC. 
When a robot $r$ is enabled at a configuration $\conf(t)$ in algorithm $A$, we say that \emph{$r$ is $A$-enabled at $\conf(t)$.}
Our simulating algorithm uses light with 3 colors, S(tay), M(ove), and E(nd). We use the notation $\forall col$  for a color $col$ denoting  a set of configuration such that all robots have color $col$. We also use the notation $\forall col_1,col_2$ for colors $col_1$ and $col_2$ denoting a set of configurations such that each robot has color $col_1$ or $col_2$ and there exists at least one robot with color $col_1$ and there exists at least one robot with color $col_2$. %
Initial configuration is in $\forall S$. This algorithm repeats a \emph{color-cycle}, that is, the transition of $\forall S\rightarrow \forall M\rightarrow \forall E$. %
When the configuration is in $\forall S$, since $\aunfair$-enabled-robots exist, some $\aunfair$-enabled-robots that become active among those execute $\aunfair$, change their colors to $M$, and move to the computed destination. While they move after changing their colors to $M$, other robots change their colors to $M$ until the configuration becomes one in $\forall M$. Note that when the configuration is in $\forall M$, some robots may be still moving. After the robots reach a configuration in $\forall M$, the robots change their colors to $E$ until the configuration becomes one in $\forall E$. In the same way, the configuration changes from a configuration in $\forall E$ to a configuration in $\forall S$. This cycle is repeated until the robots reach a configuration where no robot is $\aunfair$-enabled  and all robots are colored $S$. Each time a configuration in $\forall S$ where one or more robots are $\aunfair$-enabled is reached, at least one of them becomes active and performs $\aunfair$ observing the same configuration. In consequence, this algorithm can simulate algorithms in unfair SSYNC.

\begin{figure}[ht]
  \centering
  \includegraphics[width=0.9 \textwidth,clip]{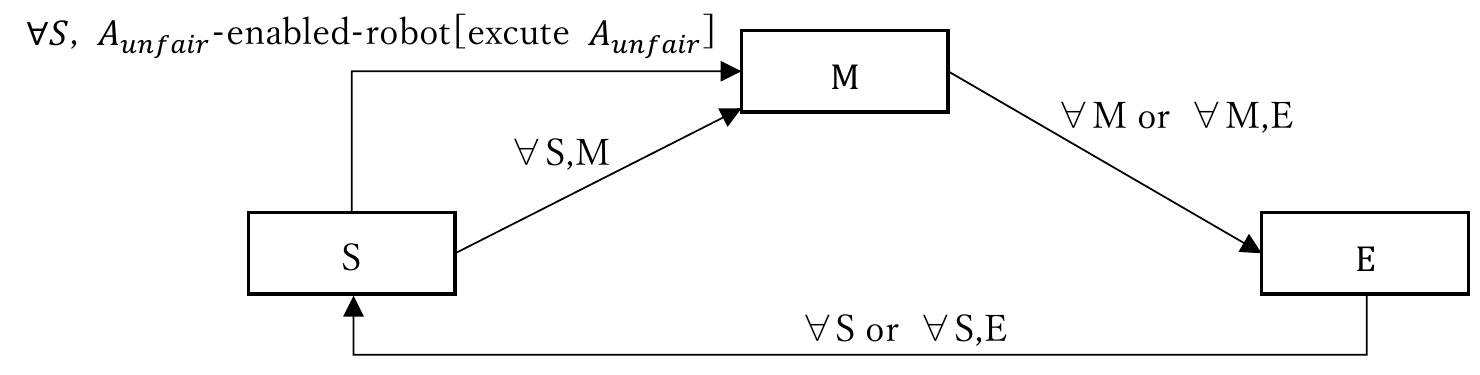}
  \caption{Transition Graph for SIM-for-Unfair($r_i$).}
  \label{fig:sim}
\end{figure}

We show that SIM-for-Unfair($r_i$) simulates $\aunfair$ correctly in ASYNC.

\begin{lemma}\label{lemma:StoM}
 Let the configuration be in $\forall S$ at time $t_S$ and let $R_{e}$ be a set of $\aunfair$-enabled-robots at $t_S$. After $t_S$, the followings hold for SIM-for-Unfair($r_i$) if $R_{e}\neq \emptyset$.
\begin{enumerate}
    \item[(1)] There is a time $t_M>t_S$ at which the configuration is in $\forall M$.
    \item[(2)] There are a time $t_f (t_S < t_f <t_M)$ and a non-empty subset $R'_e$ of $R_e$ such that all robots in $R'_e$ perform $\aunfair$ observing the same configuration and all robots in $R_e-R'_e$ do nothing between $t_S$ and $t_f$.
\end{enumerate}
\end{lemma}

\begin{proof}
Let  %
$t_f$ be a time at which the first robot changes its color from $S$ to $M$ after $t_S$. Setting $R'_e$ be a set of $\aunfair$-enabled-robots activated between $t_S$ and $t_f$,   
since $\conf(t_S)$ is not changed between $t_S$ and $t_f$, each $\aunfair$-enabled-robot in $R'_e$ observes the same configuration $\conf(t_S)$, the robot performs lines~4-6 in Algorithm~\ref{algo:SIM}. They perform $\aunfair$ with the same configuration and change their colors to $M$.
Any other robot than $R'_e$ does nothing even if it is activated between $t_S$ and $t_f$. Then (2) holds.

Activated robots observing some robot with $M$ after $t_f+1$ change their colors to $M$ (lines~8-9). This continues until $\forall M$ is satisfied and there exists a time $t_M$ stated in (1) due to the fairness of ASYNC. Note that all robots do not finish their $\mathit{LCM}$-cycle at time $t_M$.
\qed
\end{proof}

\begin{lemma}\label{lemma:MtoE}
  If the configuration is in $\forall M$ at time $t_M$, there is a time $t_E$ at which it is in $\forall E$.
\end{lemma}

\begin{proof}

If the configuration is in $\forall M$ at time $t_M$, activated robots after $t_M$ change their colors from $M$ to $E$ until the configuration in $\forall E$. Since all robots become active after $t_M$ by the fairness of ASYNC, there is a time $t_E$ at which the configuration is in $\forall E$.
\qed
\end{proof}

Note that since robots with $M$ stay when changing their colors from $M$ to $E$,
the configuration at $t_E$ is unchanged until some robots are activated after $t_E$.

The following lemma holds similarly.

\begin{lemma}\label{lemma:EtoS}
  If the configuration is in $\forall E$ at time $t_E$, there is a time $t_S$ at which it is in  to $\forall S$.
\end{lemma}

Using Lemmas~\ref{lemma:StoM}-\ref{lemma:EtoS}, we can verify that algorithm SIM-for-Unfair($r_i$) simulates $\OB$-algorithm $\aunfair$ in unfair SSYNC correctly in ASYNC with $3$ colors of $\LU$-light, and
if $\aunfair$ uses $k$ colors of $\LU$-light,  SIM-for-Unfair($r_i$) uses $3k$ colors. Then the following theorem is obtained.

\begin{theorem} \label{theorem:SIM}
  SIM-for-Unfair in ASYNC with $\LU$ of 3k colors simulates algorithms in unfair SSYNC with $\LU$ of k colors.
\end{theorem}

\subsection{Gathering Algorithm with Simulation}

In order to show that an algorithm $A$ works in unfair SSYNC, we use a concept of  potential function for $A$, which represents how close current configuration is to the final configuration.   \emph{A potential function} $f_A$ for algorithm $A$ is a function from time $t$ ($\in \N$) to feature value obtained from configuration $\conf(t)$ for the algorithm $A$, which is taken from a total ordered  set. If the potential function $f_A$ for algorithm $A$ working in SSYNC is monotonically decreasing, that is, $f_A(t) > f_A(t+1)$ for any $t$, we can show that the algorithm $A$ can work in unfair SSYNC.  

In the subsequent subsections, we show that the Gathering algorithm with two colors of light shown in SSYNC in~\cite{TWK} can still work in unfair SSYNC by constructing potential functions.

The Gathering algorithm~\cite{TWK} is divided into two sub-algorithms. The first one (called ElectOneLDS) obtains a configuration that all robots are located on one straight line segment (called onLDS) from any initial configuration, and the second one (called $\LU$-Gather) is a Gathering algorithm from any initial  configuration of onLDS.%

We obtain an algorithm with SIM-for-Unfair, ElectOneLDS, and $\LU$-Gather by replacing the line 4
in Algorithm~\ref{algo:SIM} with the line

{\bf if not} onLDS {\bf then} ElectOneLDS($r_i$) 
{\bf else} $\LU$-Gather($r_i$).

\noindent
This algorithm
 simulates ElectOneLDS($r_i$) until onLDS is attained and once onLDS is obtained it simulates $\LU$-Gather and Gathering is completed. If ElectOneLDS and $\LU$-Gather can work in unfair SSYNC, we can show that the combined algorithm solves Gathering in ASYNC.
In the subsequent subsections, We will show that ElectOneLDS and $\LU$-Gather can still work in unfair SSYNC by constructing monotonically decreasing potential functions for the both algorithms.

\subsection{ElectOneLDS works in unfair SSYNC}

ElectOneLDS~\cite{TWK} is the algorithm producing onLDS from any initial configuration and is shown in algorithm~\ref{algo:ElectOneLDS}\footnote{The original algorithm in~\cite{TWK} is slightly different, where in the case of consecutive minimum edges (line~9) all the edges are contracted but here only the rightmost edge is contracted. This change has no effect for the correctness but is necessary to work in ASYNC.}. %
The algorithm changes the control of robots by the configuration that is symmetric or asymmetric and contractible or not. Let ${\cal CH}(\conf(t))$ be the convex hull obtained by a configuration $\conf(t)$ at time $t$. $\conf(t)$ is {\em symmetric} if all edges of the convex hull ${\cal CH}(\conf(t))$ have the same length, otherwise it is {\em asymmetric}. $\conf(t)$ is {\em contractible} 
if $\conf(t)$ satisfies the following conditions;

\begin{enumerate}
\item[(i)] $\conf(t)$ is symmetric, and all robots are on the vertices or on the center of ${\cal CH}(\conf(t))$, or
\item[(ii)] $\conf(t)$ is asymmetric, and all robots are on the vertices or edges.
\end{enumerate}

\Newcodeline
\begin{algorithm}[ht]
  \caption{ElectOneLDS($r_i$)}
  \label{algo:ElectOneLDS}
  {\footnotesize
    \begin{tabbing}
      111 \= 11 \= 11 \= 11 \= 11 \= 11 \= 11
      \= \kill
      {\em Assumptions}: chirality, non-rigid, SSYNC,  ${\cal CH}({\cal SS}_i)$ is the convex hull obtained by ${\cal SS}_i$.\crm

      \Cl  {\bf case} ${\cal CH}({\cal SS}_i)$ {\bf of } \crm

      \Cl \>  non-contractible {\bf and} symmetric: \crm
      \Cl \> \>  {\bf if} $p_i$ is not a vertex nor the center of ${\cal CH}({\cal SS}_i)$ {\bf then} $des_i \leftarrow$ the center of ${\cal CH}({\cal SS}_i)$\crm
      \Cl \>  non-contractible {\bf and} asymmetric: \crm
      \Cl \> \>  {\bf if} $p_i$ is not a vertex {\bf then} $des_i \leftarrow $ the nearest vertex to $p_i$\crm
      \Cl \>  contractible {\bf and} symmetric: \crm
      \Cl \> \>  {\bf if} $p_i$ is not the center ${\cal CH}({\cal SS}_i)$ {\bf then} $des_i \leftarrow p$ the center of ${\cal CH}({\cal SS}_i)$\crm
      \Cl \>  contractible {\bf and} asymmetric: \crm
      \Cl \> \>  {\bf if} ($p_i$ is on a minimum edge {\bf or} on the rightmost edge of consecutive minimum edges)\crm   \> \> \> \>{\bf and} ($p_i$ is not the rightmost vertex of the edge)
      {\bf then} \crm
      \Cl \> \> \>  $des_i \leftarrow$ the rightmost vertex of the edge(s)\crm

      \Cl  {\bf endcase}
    \end{tabbing}
  }
\end{algorithm}

Figure~\ref{fig:ElectOneLDS} shows the transitions between configurations.
We define a potential function $f: \N \rightarrow (\R)^2 \times \N \times (\R)^2$ as follows;
The range of $f$ is a 5-dimensional vector and $f$ is denoted by $<f_1^{area},f_2^{Cdist},f_3^{\#in},f_4^{Edist},f_5^{Vdist}>$.  The order of values of the range set is taken with the lexicographic order.

Let ${\cal CH}(\conf(t))$ be the convex hull obtained from $\conf(t)$ and its vertices are denoted as $<v_0, \ldots, v_{k-1}>$, where the vertices are located counter-clockwise and  the edge $v_0v_1$ is the longest and $v_0$ is the rightmost and topmost vertex in the global coordinate system. Let $p_i (1 \leq i \leq n)$ denote the location of $r_i$.

\begin{enumerate}
    \item The function value $f^{area}_1(t)$ is the area of ${\cal CH}(\conf(t))$. 

\item If $\conf(t)$ is symmetric and not onLDS, the function value $f{^{Cdist}_2}(t)$
is the sum of distances between the center of ${\cal CH}(\conf(t))$ (denoted as $p_c$)  and all robots' locations. Otherwise it is $0$.
That is, 
$$
  f{^{Cdist}_2}(t)= \left\{
  \begin{array}{ll}
    \sum_{i=0}^{n-1}dis(p_c,p_i) &  (\conf(t)\; is\; symmetric) \\
    0                     & (\conf(t)\; is\; asymmetric).
  \end{array}
  \right.
$$
\item If $\conf(t)$ is asymmetric, the function value $f{^{\#in}_3}(t)$ is the number of robots inside ${\cal CH}(\conf(t))$. Otherwise, it is $0$.
\item If $\conf(t)$ is asymmetric, the function value $f{^{Edist}_4}(t)$ the sum of distances traversing the vertices of ${\cal CH}(\conf(t))$ counter-clockwise from  the rightmost and topmost vertex $v_0$ to the locations of robots on ${\cal CH}(\conf(t))$.\\
Let $<v'_0,v'_1, \ldots, v'_{k'-1}>$ be a sequence of the locations of robots on ${\cal CH}(\conf(t))$ such that the locations are listed counter-clockwise from the rightmost and topmost vertex $v_0$. Note that $v'_0=v_0$.
$$
  f{^{Edist}_4}(t)= \left\{
  \begin{array}{ll}
    0                                 & (\conf(t)\; is\; symmetric)\\
    \sum_{i=0}^{k'-1}dist(v_0,v'_i) & (\conf(t)\; is\; asymmetric).
  \end{array}
  \right.
$$

\item If $\conf(t)$ is asymmetric, the function value $f{^{Vdist}_5}(t)$ is the sum of distances from all robots' location to the nearest vertex of 
${\cal CH}(\conf(t))$. That is, letting $nv_i$ be the nearest vertex of ${\cal CH}(\conf(t))$ from the location $p_i$ of  robot $r_i$.\\

$$
  f{^{Vdist}_5}(t)= \left\{
  \begin{array}{ll}
    0                      & (\conf(t)\; is\; symmetric) \\
    \sum_{i=0}^{n-1}dis(nv_i,p_i) & (\conf(t)\; is\; asymmetric).  \end{array}
  \right.
$$

\end{enumerate}

\begin{figure}[ht]
  \centering
  \includegraphics[width=0.9 \textwidth,clip]{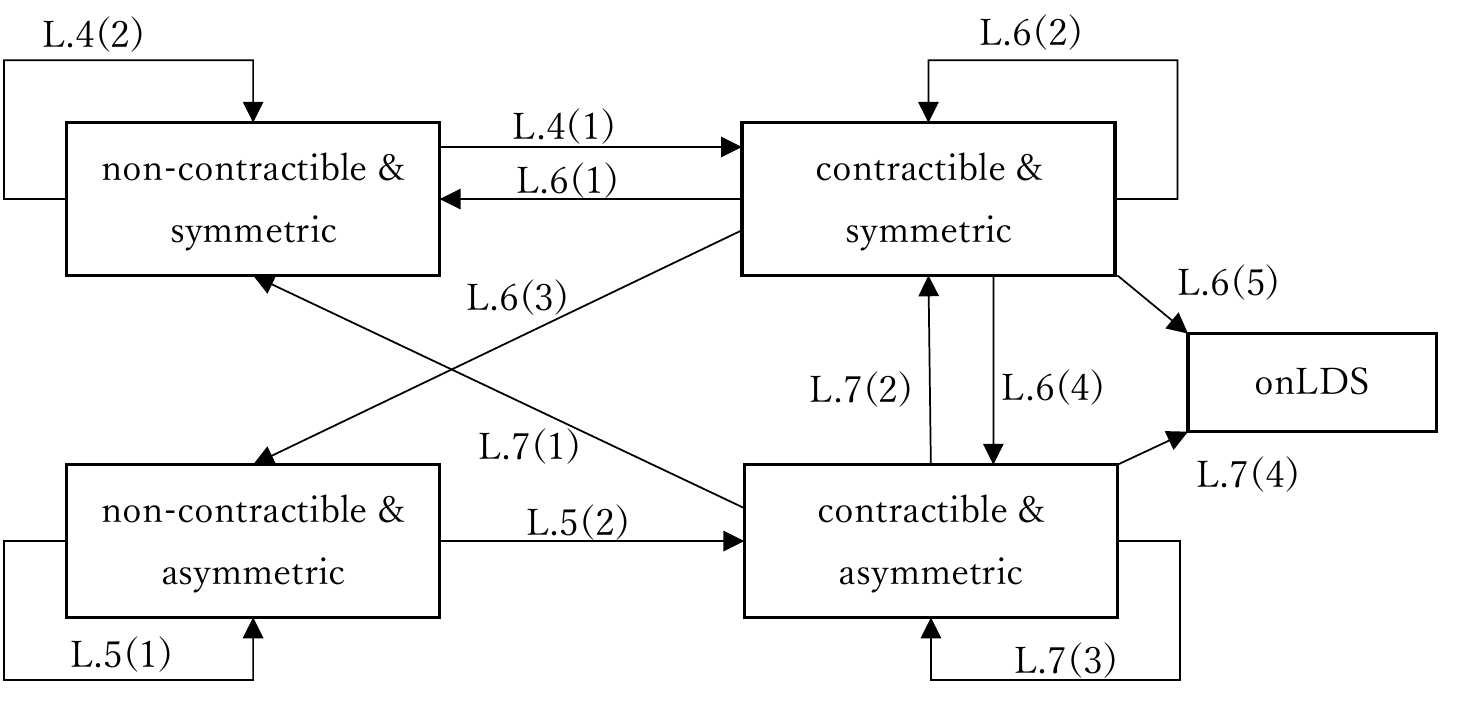}
  \caption{Transition Graph for ElectOneLDS}
  \label{fig:ElectOneLDS}
\end{figure}

\begin{table}[ht]
  \caption{In ElectOneLDS, change(dec. or inc.) of  functions' values from $f(t)$ to $f(t+1)$.}
  \begin{center}
    \begin{tabular}{|c|c|c|c|} \hline
      \multirow{2}{*}{$\conf(t)$} & \multirow{2}{*}{enabled-robots on ($\rightarrow$destination)} & \multirow{2}{*}{$\conf(t+1)$} & change of function \\
       & & & (dec.,inc.) \\ \hline \hline
      s\&nc & points except vertices($\rightarrow$the center of ${\cal CH}(\conf(t))$) & s\&nc or s\&c & ($f_2$,none) \\ \hline

      \multirow{2}{*}{a\&nc} & \multirow{2}{*}{points inside ${\cal CH}(\conf(t))$($\rightarrow$the nearest vertex)} & \multirow{2}{*}{a\&nc or a\&c} & ($f_3$\&$f_5$,$f_4$) \\
       & & & ($f_5$,none) \\ \hline

      \multirow{4}{*}{s\&c} & \multirow{4}{*}{vertices ($\rightarrow$the center of ${\cal CH}(\conf(t))$)} & \multirow{2}{*}{s\&nc, s\&c,} & %
      ($f_1$,none) \\
       & & & ($f_2$,none) \\
       & & a\&nc, a\&c & ($f_1$\&$f_2$,$f_{3-5}$)  \\
       & &or onLDS & ($f_1$\&$f_2$,none) \\ \hline

      \multirow{4}{*}{a\&c} & \multirow{4}{*}{\shortstack{ non-consecutive minimum edge or\\rightmost consecutive minimum edge\\($\rightarrow$the right endpoint of the edge)}} & s\&nc, s\&c, & ($f_1$\&$f_{3-5}$,$f_2$) \\
       & & \multirow{2}{*}{a\&c,} & %
       ($f_1$\&$f_4$,none) \\
              & & & ($f_4$,none) \\
       & & or onLDS & ($f_1$\&$f_{3-5}$,none) \\ \hline
    \end{tabular}
    {\footnotesize Abbreviated notations: "s":symmetric, "as":aymmetric, "c":contractible, "nc":non-contractible, $f_i$ omits the superscript.}
  \end{center}
\label{table:electonelds}  
\end{table}

\begin{lemma}\label{lemma:SymNonC}
  If $\conf(t)$ %
  is symmetric and non-contractible, %
  it holds that $f(t)>f(t+1)$ and
  \begin{enumerate}
      \item[(1)] $\conf(t+1)$ is symmetric and contractible, or
  \item[(2)] $\conf(t+1)$ is symmetric and non-contractible.
    \end{enumerate}
\end{lemma}

\begin{proof}

At time $t$, robots  at points except vertices and the center of ${\cal CH}(\conf(t))$ become enabled. If all enabled-robots move to the center, the configuration becomes symmetric and contractible.
Thus, although $f{^{area}_1}(t)$ is unchanged, and $f{^{Cdist}_2}(t)$ decreases.
Otherwise, the configuration remains symmetric and non-contractible. Then (2) is also holds.
\qed
\end{proof}

\begin{lemma}\label{lemma:ASymNonC}
  If $\conf(t)$ is asymmetric and non-contractible,
  it holds that $f(t)>f(t+1)$ and%
  \begin{enumerate}
\item[(1)] $\conf(t+1)$ is asymmetric and non-contractible, or
  \item[(2)] $\conf(t+1)$ is asymmetric and contractible.
    \end{enumerate}
\end{lemma}
\begin{proof}

At time $t$, a robot $r_i$ that locates at a point except vertices and edges of ${\cal CH}(\conf(t))$ becomes  enabled and moves to the nearest vertex $nv_i$.

(1) Unless all robots inside ${\cal CH}(\conf(t))$ reach $nv_i$ at $t+1$, $f_1^{area}(t)$, $f_3^{\#in}(t)$ and $f_4^{Edist}(t)$ are unchanged because ${\cal CH}(\conf(t))$ and the number of points on ${\cal CH}(\conf(t))$ do not change. In this case, since there is at least one enabled-robot that moves to $nv_i$,  $f{^{Vdist}_5}(t)$ decreases, and  $f(t)>f(t+1)$. Otherwise, since the number of points on ${\cal CH}(\conf(t))$ increases, $f_4^{Edist}$ increases. However, the number of points inside  ${\cal CH}(\conf(t))$ decreases, and so $f_3^{\#in}(t)$ and $f{^{Vdist}_5}(t)$ decrease. Therefore, comparing in lexicographic order, we have $f(t)>f(t+1)$.

(2)In this case, all robots inside ${\cal CH}(\conf(t))$ at $t$ reach to $nv_i$. Since the number of points on ${\cal CH}(\conf(t))$ increases, $f_4^{Edist}(t)$ increases. However, the number of points inside  ${\cal CH}(\conf(t))$ decreases, and $f_3^{\#in}(t)$ and $f{^{Vdist}_5}(t)$ decrease. Thus, we have $f(t)>f(t+1)$.
\qed
\end{proof}

\begin{lemma}\label{lemma:SymC}
  If $\conf(t)$ is symmetric and contractible,
  it holds that $f(t)>f(t+1)$ and%
  \begin{enumerate}
\item[(1)] $\conf(t+1)$ is symmetric and non-contractible, 
  \item[(2)] $\conf(t+1)$ is symmetric and contractible,
    \item[(3)] $\conf(t+1)$ is asymmetric and non-contractible,
      \item[(4)] $\conf(t+1)$ is asymmetric and contractible, or
        \item[(5)] $\conf(t+1)$ is onLDS.
  \end{enumerate}
\end{lemma}

\begin{proof}

At time $t$, robots that locate at vertices of ${\cal CH}(\conf(t))$ become enabled and move to the center $p_c$ of ${\cal CH}(\conf(t))$.

(1)If ${\cal CH}(\conf(t))$ does not change, there is one or more enabled-robot that move to $p_c$, and $f{^{Cdist}_2}(t)$ decreases. Otherwise, there are robots at a vertex of ${\cal CH}(\conf(t))$ at $t$ that reach at $p_c$. Then ${\cal CH}(\conf(t))$ shrinks, and $f_1^{area}(t)$ decreases.

(2)If ${\cal CH}(\conf(t))$ does not change, there is a robot that reaches $p_c$, and $f{^{Cdist}_2}(t)$ decreases. Otherwise, robots at each vertex moves by the same distance, or there are robots at a vertex of ${\cal CH}(\conf(t))$ that reach at $p_c$. Then ${\cal CH}(\conf(t))$ shrinks, and $f_1^{area}(t)$ decreases.

(3) There are robots at a vertex of ${\cal CH}(\conf(t))$ that move inside ${\cal CH}(\conf(t))$. Then ${\cal CH}(\conf(t))$ shrinks, and $f_1^{area}(t)$ decreases.

(4)There are robots at a vertex of ${\cal CH}(\conf(t))$ that move inside ${\cal CH}(\conf(t))$. Then ${\cal CH}$ shrinks, and $f_1^{area}(t)$ decreases.

(5)When all robots not on a vertex at $t$ arrived at the center $p_c$ of ${\cal CH}(\conf(t))$, or ${\cal CH}(\conf(t))$ has the regular polygon that has a diagonal through its center, and all robots not on a diagonal at $t+1$ arrived at $p_c$ of ${\cal CH}(\conf(t))$, onLDS is created at $t+1$. Since onLDS is symmetric, and $f_1^{area}(t+1)$ and $f{^{Cdist}_2}(t+1)$ is $0$, $f(t+1)=<0, 0, 0, 0, 0>$.
\qed
\end{proof}

\begin{lemma}\label{lemma:ASymC}
  If $\conf(t)$ is asymmetric and contractible,
  it holds that $f(t)>f(t+1)$ and%
  \begin{enumerate}
\item[(1)] $\conf(t+1)$ is symmetric and non-contractible, 
  \item[(2)] $\conf(t+1)$ is symmetric and contractible,
    \item[(3)] $\conf(t+1)$ is asymmetric and contractible,
      \item[(4)] $\conf(t+1)$ is onLDS.
  \end{enumerate}
\end{lemma}

\begin{figure}[ht]
  \centering
  \includegraphics[width=0.9 \textwidth,clip]{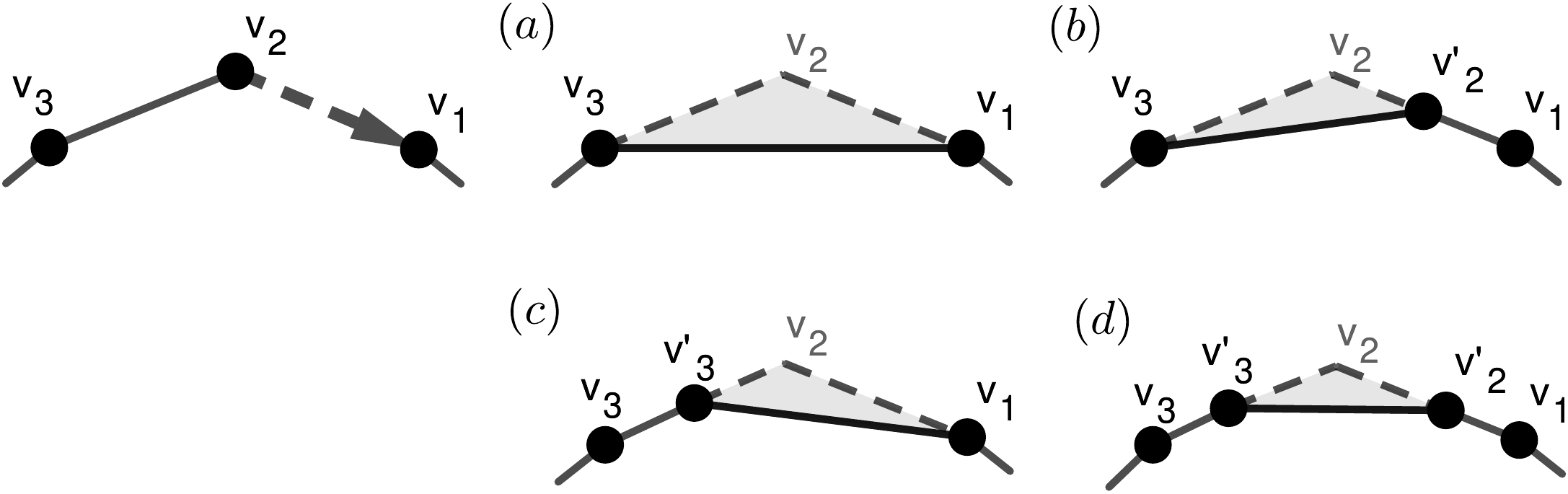}
  \caption{In cases of shrunken ${\cal CH}(\conf(t+1))$ that is asymmetric and contractible}
  \label{fig:as&c}
\end{figure}

\begin{proof}

At time $t$, robots on the rightmost edge of consecutive minimum edges or a non-consecutive minimum edge of ${\cal CH}(\conf(t))$ except the right endpoint become enabled and move to the right endpoint.
If $\conf(t+1)$ is symmetric,  since enabled-robots move along the minimum edge, the edge remains the minimum edge until the edge is contracted. Symmetric configuration is obtained by contracting the edges that is the rightmost edge of consecutive minimum edges or a non-consecutive minimum edge. Then, since ${\cal CH}(\conf(t+1))$ is obtained by removing some vertices from ${\cal CH}(\conf(t))$, ${\cal CH}(\conf(t))$ shrinks, and $f_1^{area}(t)$ decreases. This proves (1) and (2).

If $\conf(t+1)$ is asymmetric, enabled-robots move along the minimum edge. Figure~\ref{fig:as&c} shows the cases where the contracted edge is $v_1v_2$, and ${\cal CH}(\conf(t+1))$ and ${\cal CH}(\conf(t))$ are different due to all robots on $v_2$ move toward $v_1$. In (a), there are not robots on the edge $v_2v_3$, and robots on $v_2$ reach $v_1$. In(b), there are not robots on the edge $v_2v_3$, and robots on $v_2$ move toward $v_1$ and stopped on the edge. In (c), there is a robot on the edge $v_2v_3$, and robots on $v_2$ reach $v_1$. In (d), there is a robot on the edge $v_2v_3$, and robots on $v_2$ move toward $v_1$ and stopped on the edge. In these cases, ${\cal CH}(\conf(t))$ shrinks, and $f_1^{area}(t)$ and $f_4^{Edist}(t)$ decrease.Even if ${\cal CH}(\conf(t))$ does not change, at least one enabled-robot moves along the right vertex, and $f_4^{Edist}(t)$ decreases. Therefore, (3) holds.

When $\conf(t+1)$ is onLDS, $\conf(t)$ is either the triangle that has one or two minimum edges or a rectangle except square. If $\conf(t)$ is the triangle that has one or two minimum edges, and there are no robots on the left-hand edge of contracted edge at $t+1$, onLDS is made by contracting one edge. Note that in the case of the triangle that has two minimum edges, only the rightmost edge of the two is contracted. Since onLDS is symmetric, and $f_1^{area}(t)$ and $f{^{Cdist}_2}(t)$ are $0$, we have $f(t+1)=<0, 0, 0, 0, 0>$. If $\conf(t)$ is a rectangle except square, onLDS is made by contracting the two minimum edge. Since onLDS is symmetric, and $f_1^{area}(t+1)$ and $f{^{Cdist}_2}(t+1)$ are 0, $f(t+1)=<0, 0, 0, 0, 0>$. Therefore, (4) holds.
\qed
\end{proof}

Table~\ref{table:electonelds} summarizes Lemmas~\ref{lemma:SymNonC}-\ref{lemma:ASymC}. Elements of the potential function on $\R$ decrease depending on $\delta$ even if the area of ${\cal CH}$, because robots under non-rigid movement can move at least a distance $\delta$. Therefore, the next theorem holds.

\begin{theorem} \label{theorem:SIMonLDS}
  The potential function $f$ for ElectOneLDS is monotonically decreasing.
\end{theorem}

\color{black}

\subsection{$\LU$-Gather works in unfair SSYNC}

$\LU$-Gather~\cite{TWK} is the Gathering algorithm where it starts with any onLDS configuration and it uses $2$ $\LU$-colors ($A$ and $B$) and works in SSYNC. $\LU$-Gather is shown in Algorithm~\ref{algo:FLG-withcc} and the transition between color-\\configurations is depicted in Figure~\ref{fig:FLG}. 
The outline of Algorithm~\ref{algo:FLG-withcc} is as follows; 

The initial configuration of $\LU$-Gather is onLDS  and the algorithm makes  $AA$ via $AA^+A$ and robots at endpoints move to the midpoint after changing its colors to $B$. Then the color-configuration becomes $AB_mA$ or $AB^+A$, $AB^*B$ or $BB^*A$, or $BB^*B$ or $BB_mB$.
\begin{enumerate}
    \item[(1)] For the case that the number of point with $A$ is one ($AB^*B$ or $BB^*A$), the points with $A$ is the Gathering point and so robots with $A$ stay and robots with $B$ move to the point with $A$ and Gathering is attained. 
\item[(2)] For the case that the number of points with $A$ is zero ($BB^*B$), robots at endpoints change their colors to $A$, the color configuration is $AB^*A$ ($AA$, $AB_mA$ or $AB^+A$). The first case is the same as the initial color configuration, but the distance between endpoints is decreased by at least $2\delta$, and the second case is treated in (3).
\item[(3)] For the case that the number of points with $A$ is two ($AB_mA$ or $AB^+A$), after making $AB_mA$ from $AB^+A$, robots at endpoints move to the midpoint after changing their colors to $B$. Then Gathering is attained, or  the color configuration becomes one same as the first transition from the initial configuration but the distance between endpoints is decreased by at least $2\delta$. If the distance between endpoints is at most $2\delta$, since robots at endpoints can reach the midpoint this repetition is finished in finite times and Gathering is attained.
\end{enumerate}
\begin{figure}[ht]
  \centering
  \includegraphics[width=0.9 \textwidth,clip]{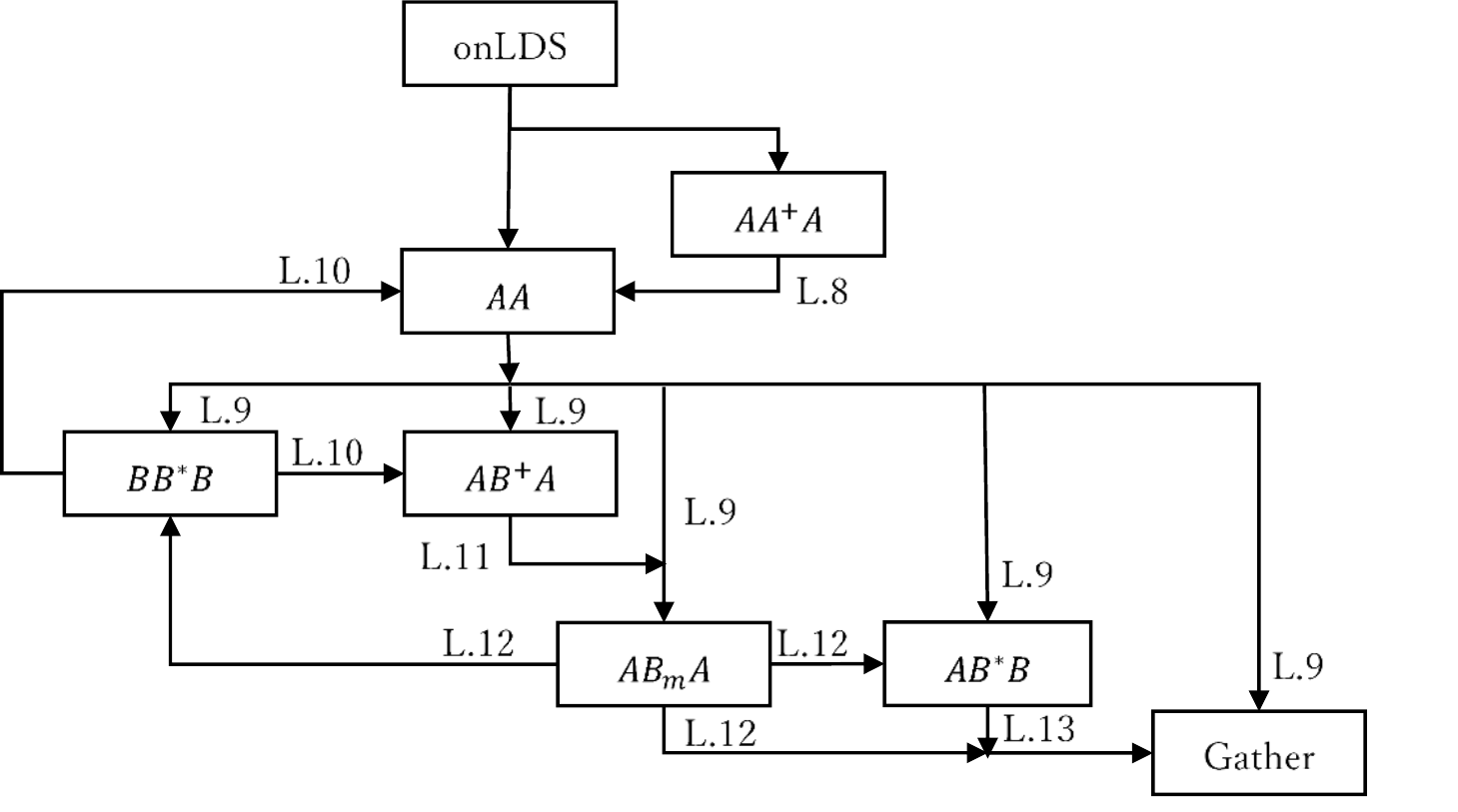}
  \caption{Transition Graph for $\LU$-Gather}%
  \label{fig:FLG}
\end{figure}
\Newcodeline
\begin{algorithm}[ht]
\caption{$\LU$-Gather($r_i$)}
\label{algo:FLG-withcc}
{\footnotesize
\begin{tabbing}
111 \= 11 \= 11 \= 11 \= 11 \= 11 \= 11 \= \kill
{\em Assumptions}: $\LU$, $2$ colors($A$ and $B$), non-rigid, SSYNC, initially $\ell_i=A$, \crm
      \>\>\>\> let $cc({\cal SS}_i)$ denote  color-configuration of snapshot ${\cal SS}_i$ of $r_i$, \crm
      \>\>\>\> initial configuration is onLDS.\crm

\Cl \> {\bf case} $cc({\cal SS}_i)$  {\bf of } \crm

\Cl \> $\in \forall A$: \crm
\Cl \> \> {\bf if} $cc({\cal SS}_i)=A$ {\bf then} $des_i \leftarrow p_i$ // gathered!\crm
\Cl \> \> {\bf else if} $cc({\cal SS}_i)=AA$ {\bf then}  $\ell_i \leftarrow B$; $des_i \leftarrow (p_n + p_f)/2$\crm
\Cl \> \> {\bf else} //$cc({\cal SS}_i)=AA^+A$\crm
\Cl \> \> \> {\bf if} $(p_i = p_n)$ {\bf then} $des_i \leftarrow p_i$ //If $r_i$ is at either endpoint, it does not move.\crm 
\Cl \> \> \> {\bf else} $des_i \leftarrow p_n$//To make $AA$ from $AA^+A$\crm%
\Cl \> $\in \forall B$: \crm
\Cl \>\> {\bf if} $cc({\cal SS}_i)=BB^*B$ {\bf and} $(p_i = p_n)$ {\bf then} $\ell_i \leftarrow A$//To make $AB^*A$ from $BB^*B$\crm
\Cl \> \>{\bf else}  do nothing\crm
\Cl \> $\in \forall A, B$: \crm
\Cl \> \> {\bf if} $\ell_i = A$ {\bf then}\crm
\Cl \> \> \>{\bf if} $cc({\cal SS}_i)=AB^*B$ (or $BB^*A$)  {\bf then} $des_i \leftarrow p_i$ // stay\crm
\Cl \> \>  \>{\bf else if} $cc({\cal SS}_i)=AB_mA$ {\bf then}\crm
\Cl \> \> \>\> $\ell_i \leftarrow B$ \crm
\Cl \> \> \> \>$des_i \leftarrow  (p_n + p_f)/2$\crm
\Cl \> \> {\bf else} // $\ell_i = B$\crm 
\Cl \> \> \>{\bf if} $cc({\cal SS}_i)=AB^*B$ (or $BB^*A$) {\bf then} $des_i \leftarrow$ point with $A$ \crm
\Cl \> \> \>{\bf else} $des_i \leftarrow (p_n + p_f)/2$ //$cc({\cal SS}_i)=AB^+A$ \crm 

\Cl \> {\bf endcase} 
\end{tabbing}
}
\end{algorithm}

We define a potential function  $g: \N \rightarrow (\R\cup\{\infty\})^3 \times \N \times \R$ as follows;
The range of $g$ is a 5-dimensional vector and $g$ is denoted by $<g_1^{Adist},g_2^{Edist},g_3^{Mdist},g_4^{\#B},g_5^{\mathit{NEdist}}>$.

\begin{enumerate}
    \item 
If $\#_A(\conf(t)) =1$ then the function vaule $g{^{Adist}_1}(t)$ is the sum of distances between the position with color $A$ (denoted as $p_A$) and all robot's position, otherwise, it is $\infty$.
That is,
$$
  g{^{Adist}_1}(t)= \left\{
  \begin{array}{ll}
    \sum_{i=0}^{n-1} dis(p_A,p_i) & (\#_A(\conf(t))=1)     \\
    \infty                & (\#_A(\conf(t))> 1).
  \end{array}
  \right.
$$
\item
In the case that $\#_A(\conf(t)) =0$ or $2$,  the function value $ g{^{Edist}_2}(t)$ is the distance between two endpoints.
If  $\#_A(\conf(t)) =1$, then it is $0$, otherwise it is $\infty$.%
$$
  g{^{Edist}_2}(t)= \left\{
  \begin{array}{lll}
    0         & (\#_A(\conf(t)) =1)                               \\
    dis(\conf(t)) & (\#_A(\conf(t)) =0,2) \\
    \infty    & (\#_A(\conf(t)) \geq 3).
  \end{array}
  \right.
$$
\item In the case that $\#_A(\conf(t)) =0$ or $2$,  the function value 
$g{^{Mdist}_3}(t)$ is the sum of distances between every robot's position and the midpoint of LDS (denoted as $p_m$). 
If  $\#_A(\conf(t)) =1$, then it is $0$, otherwise it is $\infty$. That is,
$$
  g{^{Mdist}_3}(t)= \left\{
  \begin{array}{lll}
    0                     & (\#_A(\conf(t)) =1)                                \\
    \sum_{i=0}^{n-1}dis(p_m,p_i) & (\#_A(\conf(t)) =0,2) \\
    \infty                & (\#_A(\conf(t)) \geq 3).
  \end{array}
  \right.
$$
\item In the case that $\#_A(\conf(t)) =0$ or $2$,  the function value 
$g{^{\#B}_4}(t)$ is the number of robots with color $B$. If  $\#_A(\conf(t)) =1$, then it is $0$, otherwise it is $\infty$. Let $R_B=\{r|r$ has color $B \}$.\\
$$
  g{^{\#B}_4}(t)= \left\{
  \begin{array}{lll}
    0      & (\#_A(\conf(t)) =1)                                \\
    |R_B|  & (\#_A(\conf(t)) =0,2) \\
    \infty & (\#_A(\conf(t)) \geq 3).
  \end{array}
  \right.
$$
\item If $\#_A(\conf(t))  \leq 3$, the function value
$g{^{\mathit{NEdist}}_5}(t)$ is the sum of distances between every robot's position $p_i$ and the nearest endpoint to  $p_i$ (denoted as $np_i$), otherwise it is $0$. %
$$
  g{^{\mathit{NEdist}}_5}(t)= \left\{
  \begin{array}{ll}
    0                     & (\#_A(\conf(t))  > 3)     \\
    \sum_{i=0}^{n-1}dis(np_i,p_i) & (\#_A(\conf(t))  \leq 3).
  \end{array}
  \right.
$$
\end{enumerate}

Note that following the definition of $g$,\\
if $|{\cal P}_A({\conf(t)})|=1$, then
$$g(\conf(t))=<\sum_{i=1}^n|p_A-p_i|, 0, 0, 0, 0>\conf(t),$$
if $|{\cal P}_A({\conf(t)})|=0$, or $|{\cal P}_A({\conf(t)})|=2$, then
$$g(\conf(t))=<\infty,|p_r-p_l|, \sum_{i=1}^n|p_i-p_m|, |R_B|, 0>\conf(t)$$,
if $|{\cal P}_A({\conf(t)})|\geq 3$, then
$$g(\conf(t))=<\infty, \infty, \infty, \infty, \sum_{i=1}^n|p_n-p_i|>\conf(t).$$

\begin{table}[ht]
  \caption{In $\LU$-Gather, change(dec. or inc.) of  functions' values from $g(t)$ to $g(t+1)$.}
  \begin{center}
    \begin{tabular}{|c|c|c|c|} \hline
      \multirow{2}{*}{$cc(\conf(t))$} & \multirow{2}{*}{enabled-robots on($\rightarrow$destination\&color)} & \multirow{2}{*}{$cc(\conf(t+1))$} & change of function\\
       & & & (dec.,inc.)\\
      \hline \hline
      \multirow{2}{*}{$AA^+A$} & \multirow{2}{*}{}points except endpoints & $AA^+A$ & ($g_5$,none)\\
       &($\rightarrow$the nearest $p_e$\&unchanged) & or $AA$ & ($g_{2-5}$,none) \\ \hline

      \multirow{4}{*}{$AA$} & \multirow{4}{*}{ endpoints($\rightarrow$the midpoint\&$B$)} & $AB^+A$, $AB_mA$, & ($g_3$,$g_4$) \\
       & & $AB^*B$, & ($g_{1-4}$,none) \\
              & & $BB^*B$, & ($g_2$,$g_4$) \\
       & & or Gather & ($g_2$\&$g_3$,$g_4$) \\\hline

      \multirow{2}{*}{$BB^*B$}  & \multirow{2}{*}{ endpoints($\rightarrow$stay\&$A$)}     & $AB^*A$& ($g_4$,none) \\
       & & or $AB^*B$ & ($g_{1-4}$,none) \\ \hline
\multirow{2}{*}{$AB^+A$}  & points except endpoints & $AB^+A$& \multirow{2}{*}{($g_3$,none)} \\
       &($\rightarrow$the midpoint\&unchanged) & or $AB_mB$ & \\ \hline

      \multirow{4}{*}{$AB_mA$} & \multirow{4}{*}{ endpoints($\rightarrow$the midpoint\&$B$)} & $AB^+A$, & ($g_3$,$g_4$)\\
       & & $AB^+B$, & ($g_{1-4}$,none)\\
       & & $BB^*B$, & ($g_2$,$g_4$) \\
       & &or Gather & ($g_2$\&$g_3$,$g_4$) \\ \hline
\multirow{2}{*}{$AB^*B$}  & points except endpoint with $A$  & $AB^*B$& \multirow{2}{*}{($g_1$,none)} \\
       &($\rightarrow$the endpoint with $A$\&unchanged) & or Gather & \\ \hline
    \end{tabular}
    {\footnotesize Function $g_i$ omits the superscript.}
  \end{center}\label{table:LU-Gather}
\end{table}

\begin{lemma}\label{lemma:AAA}
  If $cc(\conf(t))=AA^+A$, it holds that  $g(t)>g(t+1)$.
\end{lemma}

\begin{proof}
At time $t$, robots locating inside onLDS become enabled and move to the nearest endpoint $p_n$. In the case of $\#_A(\conf(t)) \geq 3$, there is at least one enabled-robot that moves to $p_n$, and $g{^{\mathit{NEdist}}_5}(t)$ decreases. In the case of $\#_A(\conf(t))= 2$, $g{^{Edist}_2}(t)$ decreases from $\infty$ to the distance of the segment. Therefore, this lemma holds.
\qed
\end{proof}

\begin{lemma}\label{lemma:AA}
  If $cc(\conf(t))=AA$, it holds that $g(t)>g(t+1)$.
\end{lemma}

\begin{proof}

At time $t$, robots at the endpoints become enabled and move to the midpoint $p_m$ of the two endpoints. In the case that $\#_A(\conf(t+1))= 2$, the endpoints do not change, and at least one enabled robot changes its color to $B$ and moves to $p_m$. Thus, $g{^{\#B}_4}(t)$ increases, but $g{^{Mdist}_3}(t)$ decreases. 
In the case that $\#_A(\conf(t+1))= 1$, $g{^{Adist}_1}(t)$ decreases from $\infty$ to $\sum_{i=0}^{n-1}dis(p_a,p_i)$.
In the case that $\#_A(\conf(t+1))=0$ and not Gather, all robots change their colors to $B$ and move to $p_m$. Then $g{^{\#B}_4}(t)$ increase, and $g{^{Mdist}_3}(t)$ may also increase, but $g{^{Edist}_2}(t)$ decreases. In the case that $\#_A(\conf(t+1))=0$ and Gather, since $g{^{Edist}_2}(t+1)$, $g{^{Mdist}_3}(t+1)$, and $g{^{\#B}_4}(t+1)$ are 0, $g(t+1)=<\infty, 0,0,0,0>$. Therefore, $g(t)>g(t+1)$.
\qed
\end{proof}

\begin{lemma}\label{lemma:BBB}
  If $cc(\conf(t))=BB^*B$, it holds that $g(t)>g(t+1)$.
  
\end{lemma}

\begin{proof}
At time $t$, robots at the endpoints become enabled and change their colors to $A$. 
If $\#_A(\conf(t+1))= 2$, clearly the number of robots with $A$ increases, and only $g{^{\#B}_4}(t)$ decreases. 
Otherwise $\#_A(\conf(t+1))= 1$, $g{^{Adist}_1}(t)$ decreases from $\infty$ to $\sum_{i=0}^{n-1}dis(p_a,p_i)$. Therefore, $g(t)>g(t)$.
\qed
\end{proof}

\begin{lemma}\label{lemma:ABpA}
  If $cc(\conf(t))=AB^+A$, it holds that that $g(t)>g(t+1)$.
\end{lemma}

\begin{proof}
Robots with $B$ at points except $p_m$ and the endpoints become enabled and move to $p_m$. Since robots at the endpoints are unchanged, it is enough to consider the case where $\#_A(\conf(t))=2$. Since there is at least one enabled-robot that moves to $p_m$, only $g{^{Mdist}_3}(t)$ decreases.
\qed
\end{proof}

\begin{lemma}\label{lemma:ABA}
  If $cc(\conf(t))=AB_mA$, it holds that that $g(t)>g(t+1)$.
\end{lemma}
\begin{proof}
Robots with $A$ at the endpoints become enabled, change their colors to $B$, and move to $p_m$. In the case of $\#_A(\conf(t))=2$, robots at the endpoints are unchanged, and there is at least one enabled-robot that changes its color $B$ and moves. Since $g{^{\#B}_4}(t)$ increases, but $g{^{Mdist}_3}(t)$ decreases.

In the case of $\#_A(\conf(t))=1$, $g{^{Adist}_1}(t)$ decreases from $\infty$ to $\sum_{i=0}^{n-1}dis(p_a,p_i)$. 
In the case of $\#_A(\conf(t))=0$ and not Gather, all robots at the endpoints change their colors $B$ and move to $p_m$. Then although $g{^{\#B}_4}(t)$ increases and $g{^{Mdist}_3}(t)$ may also increases,  $g{^{Edist}_2}(t)$ decreases. In the case that $\#_A(\conf(t+1))=0$ and Gather, since $g{^{Edist}_2}(t+1)$, $g{^{Mdist}_3}(t+1)$, and $g{^{\#B}_4}(t+1)$ are 0, $g(t+1)=<\infty, 0,0,0,0>$. Therefore, $g(t)>g(t+1)$.
\qed
\end{proof}

\begin{lemma}\label{lemma:ABB}
  If $cc(\conf(t))=AB^*B$, it holds that that $g(t)>g(t+1)$.
\end{lemma}

\begin{proof}
At time $t$, robots with color $B$ become enabled and move to the point  with color $A$ (denoted as $p_A$). Then it is enough to consider the case where $\#_A(\conf(t))=1$. Since there is at least one enabled-robot that moves to $p_A$, $g{^{Adist}_1}(t)$ decreases.
\qed
\end{proof}

Table~\ref{table:LU-Gather} summarizes Lemmas~\ref{lemma:AAA}-\ref{lemma:ABB}, which follow the next theorem.

\begin{theorem} \label{theorem:SIMgather}
  The potential function $g$ for $\LU$-Gather is monotonically decreasing.
\end{theorem}

We obtain the following theorem by Theorems~\ref{theorem:SIM}-\ref{theorem:SIMgather} 
\begin{theorem} \label{theorem:Gathering6col}
  Gathering can be solved in ASYNC by $\LU$ robots having 6 colors under non-rigid movement, and agreement of chirality.
\end{theorem}

In the next section, we can reduce the number of colors to three by construct a Gathering algorithm starting from OnLDS working in ASYNC with three colors.

\section{Gathering Algorithm in ASYNC with 3 colors}
\label{sec:GatheringAlgorithm}

In this section, we give a Gathering algorithm called 3-color-Gather-in-ASYNC working in ASYNC with $3$ colors. The pseudocode is shown in Algorithm~\ref{algo:3colorGather}.
This algorithm consists of two algorithms, where one is to make onLDS and uses the simulation of ElectOneLDS (SIM-for-Unfair[ElectOneLDS]), and the other is a Gathering algorithm from onLDS and does not use the simulation and is newly developed (called $\LU$-Gather-in-ASYNC). As we will show in Corollary~\ref{corollary:destWhenSwitch}, in SIM-for-Unfair[ElectOneLDS], once a configuration becomes onLDS, it remains onLDS forever. Therefore, the algorithm works in ASYNC with $3$ colors and therefore  3-color-Gather-in-ASYNC attains Gathering in ASYNC with $3$ colors.

\Newcodeline
\begin{algorithm}[h]
  \caption{3-color-Gather-in-ASYNC($r_i$)}
  \label{algo:3colorGather}
  {\footnotesize
    \begin{tabbing}
      111 \= 11 \= 11 \= 11 \= 11 \= 11 \= 11
      \= \kill
      {\em Assumptions}: non-rigid, ASYNC,  \crm
      Subroutine: SIM-for-Unfair($r_i$), ElectOneLDS($r_i$), $\LU$-Gather-in-ASYNC($r_i$); \crm
      \Cl \> {\bf if not} onLDS {\bf then} SIM-for-Unfair($r_i$)[ElectOneLDS] \crm
      \Cl \> {\bf else} $\LU$-Gather-in-ASYNC($r_i$) 
    \end{tabbing}
  }
\end{algorithm}

\subsection{Configurations becoming onLDS}

In 3-color-Gather-in-ASYNC, it is switched to $\LU$-Gather-in-ASYNC from the simulation when the configuration becomes onLDS.
We consider configurations which become onLDS when  ElectOneLDS is simulated by SIM-for-Unfair.

Since we are concerned with ASYNC, $\conf(t)$ contains moving robots\footnote{Robots having performed $\comp$ and not finishing $\move$ yet.} and/or robots having performed $\look$ but not performing $\comp$. The former robots are said to be in \emph{pending move} at $\conf(t)$ and the latter robots are said to be in \emph{pending color} at $\conf(t)$~\cite{DHTW}.
Then the following notations are introduced in color-configurations.
In factor $f$ of a color-configuration for $\conf(t)$, if some robots have the possibility to be in pending move or pending color at a position represented by $f$, the factor is denoted by $f[pm]$ and $f[pc\rightarrow \alpha]$, respectively. %
If there is possibility of robots being in pending move and in pending color, the factor is denoted by $f[pm,pc]$, where $pc\rightarrow \alpha$ shows that the color is changed to $\alpha$ when performing \emph{Compute}.
When robots in pending move with color $\alpha$ move to the destination $d$ in the factor $\alpha[pm]$, we say that \emph{$\alpha[pm]$ has destination $d$}.

\begin{lemma}\label{lemma:ColorOnLDS}
  If the configuration becomes onLDS at t when SIM-for-Unfair simulates ElectOneLDS. It holds that
  \begin{enumerate}
      \item[(1)] $cc(\conf(t))=SS^*S$,
  \item[(2)]$cc(\conf(t))=(S|S[pc\rightarrow M]|M|M[pm])(S|S[pc\rightarrow M]|M|M[pm])^*(S|S[pc\rightarrow M]|M|M[pm])$ with at least one $M$, 
  \item[(3)] $cc(\conf(t))=(M|M[pm,pc\rightarrow E]|E)(M|M[pm,pc\rightarrow E]|E)^*(M|M[pm,pc\rightarrow E]|E)$ with at least one $M$,
  \item[(4)] $cc(\conf(t))=(S|S[pc\rightarrow M]|M)$ with at least one $M$, or
  \item[(5)] $cc(\conf(t))=(M|M[pc\rightarrow E]|E)$ with at least one $M$.
    \end{enumerate}
    In (2) and (3), all $M[pm]$ has a destination of a point on the straight line through onLDS in $\conf(t)$.
\end{lemma}

\begin{proof}
If the initial configuration is already onLDS, trivially (1) holds.
  We first prove about $cc(\conf(t))$. If $\conf(t')$ is not onLDS for a time $t'<t$, at least one enabled-robot has to reach its destination to become onLDS. Since it has color $M$, at least one $M$ exists in $\conf(t)$. If there is a robot that has never performed \emph{Compute} until t, $\conf(t)$ is in the course of transition from configuration in $\forall S$ to configuration in $\forall M$. It is possible that there is a robot with $S$ looks a configuration in $\forall S$ before $t$, or a robot with $M$ is moving toward its destination of a point on onLDS. Hence, $cc(\conf(t))$ is $(S|S[pc\rightarrow M]|M|M[pm])(S|S[pc\rightarrow M]|M|M[pm])^*(S|S[pc\rightarrow M]|M[pm])$ with at least one $M$ or $(S|S[pc\rightarrow M]|M)$ with at least one $M$. 
  Otherwise, the configuration becomes a configuration in $\forall M$ at say, $t_{\forall M}$ before $t$. Then robots except $\aunfair$-enabled-robots changed their colors to $M$ at $t_{M}(t_{M}<t_{\forall M}<t)$, where $\conf(t_M)$ has color $M$. If a robot with $M$ performs \emph{Look} between $t_{\forall M}$ to $t$, it tries to change its color to $E$. Since $\conf(t)$ is in the course of transition from configuration in $\forall M$ to configuration in $\forall E$, and then $cc(\conf(t))=(M|M[pm,pc\rightarrow E]|E)(M|M[pm,pc\rightarrow E]|E)^*(M|M[pm,pc\rightarrow E]|E)$ with at least one $M$ or $(M|M[pc\rightarrow E]|E)$ with at least one $M$.

  Next, we show that destinations of robots in pending move with $M$ are in onLDS. Let $t$ be a time at which a configuration is onLDS for the first time, and let the time  $t_B=max\{t_S|\conf(t_S)\in \forall S, t_S<t\}$. At $t_B$, configurations are only the following cases. 
  \begin{enumerate}
      \item[(a)] $\conf(t_B)$ is asymmetric and contractible, and is a triangle that has one or two minimum edge, 
  \item[(b)]  $\conf(t_B)$ is asymmetric and contractible, and is a rectangle, or
  \item[(c)] $\conf(t_B)$ is symmetric and contractible.
\end{enumerate}

  For (a), if ${\cal CH}(\conf(t_B))$ has exactly one minimum edge, robots on the edge except the rightmost vertex become enabled at $t_B$. When all of them arrived at the rightmost vertex, onLDS is made. Hence, robots with $M$ must be stopped at $t$. If ${\cal CH}(\conf(t_B))$ has the two minimum edges, robots on a right-hand edge of the two and not on the rightmost vertex become enabled at $t_B$. When they arrived at the rightmost vertex, onLDS is made. Hence, robots with $M$ must be also stopped.
  
   For (b), robots on the minimum edges become enabled at $t_B$. When they all arrived at the rightmost vertex, onLDS is made. Hence, robots with $M$ must be stopped.
   
  For (c), there are two cases where LDS is obtained at $t$. One is a line segment whose endpoints on a diagonal through the center of ${\cal CH}(\conf(t_B))$ (denoted as $p_c$), and the other is one connecting $p_c$ and a point on a diagonal through $p_c$. Let $x$ and $y$ be the diagonal for the former case and let $z$ be the endpoint except $p_c$ for the latter case. In the former case all robots on vertices other than $x$ and $y$ have reached at $p_c$ and are stopped at $t$. Robots at $x$ or $y$ are in pending color or pending move and their destination is $p_c$. Thus, their destinations are in onLDS.
  The latter case can be shown similarly. 
\qed
\end{proof}

\begin{corollary}\label{corollary:destWhenSwitch}
In SIM-for-Unfair[ElectOneLDS], if the configuration becomes \\onLDS from non-onLDS at time $t$, destination of any moving robot at $t$ is a point on the straight line through onLDS in $\conf(t)$.
\end{corollary}

We will show that $\LU$-Gather-in-ASYNC can work from the configurations shown in Lemma~\ref{lemma:ColorOnLDS}.

\begin{figure}[h]
  \centering
  \includegraphics[width=0.9 \textwidth,clip]{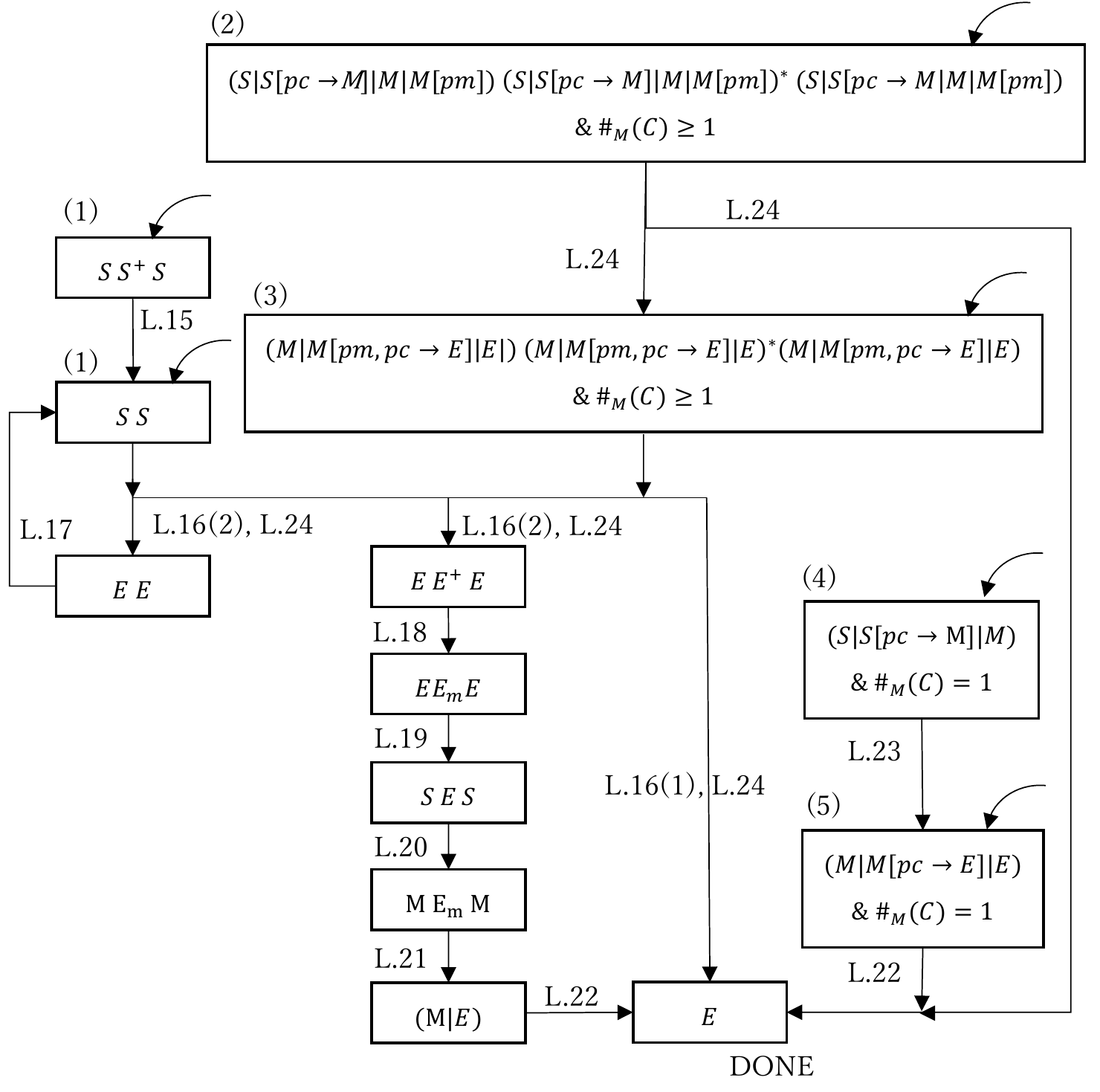}
  \caption{Transition Graph for $\LU$-Gather-in-ASYNC from $\conf$ in Lemma~\ref{lemma:ColorOnLDS}.}
  \label{fig:FLGA}
\end{figure}

\subsection{Correctness of $\LU$-Gather-in-ASYNC}

\Newcodeline
\begin{algorithm}
  \caption{$\LU$-Gather-in-ASYNC($r_i$)}
  \label{algo:FLGA}
  {\footnotesize
    {\footnotesize 
    \begin{tabbing}
      111 \= 11 \= 11 \= 11 \= 11 \= 11 \= 11 \= \kill
      {\em Assumptions}: non-rigid, $\LU$, $3$ colors($S$,$M$ and $E$). \crm
      {\em Input}: configuration onLDS and configuration satisfying Lemma~\ref{lemma:ColorOnLDS}.\crm %
      \Cl \> (Let $p_n$ be the nearest endpoint to $p_i$, and let $p_f$ be the furthest endpoint to $p_i$)\crm
      \Cl \> {\bf case} $cc({\cal SS}_i)$ {\bf of } \crm
      \Cl \> $\in \forall S$:\crm
      \Cl \> \> {\bf if} $cc({\cal SS}_i)=SS$ {\bf then}// $\rightarrow SMS$, $SM^+S$, $MM^*M$, $SMM$, or $SM^*M$ \crm
      \Cl \> \> \> $l_i \leftarrow M$\crm
      \Cl \> \> \> $des_i \leftarrow (p_n+p_f)/2$\crm
      \Cl \> \> {\bf else} //$cc({\cal SS}_i)=SS^+S$ \crm
      \Cl \> \> \> {\bf if} $p_i\neq p_n$ {\bf then} $des_i\leftarrow p_n$ //$\rightarrow SS$ \crm
      \Cl \> $\in \forall S,M$: \crm
      \Cl \> \> {\bf if} $cc({\cal SS}_i)=M^+(S|M)M^* $, $M^*(S|M)M^+$, or $(S|M)$ {\bf then}\crm
      \Cl \> \> \> {\bf if} $l_i=S$ {\bf then} $l_i \leftarrow E$ //$M^+(S|M)M^*\rightarrow M^+EM^*$ or $(S|M)\rightarrow (M|E)$ \crm
      \Cl \> \> {\bf else if} $cc({\cal SS}_i)=(S|M)M^*(S|M)$ {\bf and} $\#_S({\cal SS}_i)=2$ {\bf and} $l_i=S$ {\bf then} \crm
      \Cl \> \> \> $l_i \leftarrow M$ //$\rightarrow SM^*M$ or $MM^*M$ \crm
      \Cl \> \> \> $des_i \leftarrow (p_n+p_f)/2$ \crm
      \Cl \> \> {\bf else if} $\#_S({\cal SS}_i)\geq 2$ {\bf and} $l_i=S$ {\bf then} $l_i\leftarrow M$ //The number of S decreases. \crm
      \Cl \> $\in \forall S,E$: \crm
      \Cl \> \> {\bf if} $cc({\cal SS}_i)=(S|E)(S|E)$ {\bf and} $l_i=E$ {\bf then} $l_i \leftarrow S$ //$\rightarrow SS $ \crm
      \Cl \> \> {\bf else} //$(S|E)E(S|E)$ \crm
      \Cl \> \> \> {\bf if} $cc({\cal SS}_i)=(S|E)E(S|E)$ {\bf and} $\#_E({\cal SS}_i)>1$ {\bf and} $p_i=p_n$ {\bf and} $l_i=E$ {\bf then} \crm
      \Cl \> \> \> \> $l_i \leftarrow S$ //$\rightarrow SES$ \crm
      \Cl \> \> \> {\bf else if} $cc({\cal SS}_i)=SES$ {\bf and} $l_i=S$ {\bf then} //$\rightarrow (S|M)E(S|M)$ \crm
      \Cl \> \> \> \> $l_i \leftarrow M$\crm
      \Cl \> $\in \forall M$: \crm
      \Cl \> \> $l_i \leftarrow E$ //$MM^*M \rightarrow (M|E)(M|E)^*(M|E)$ \crm
      \Cl \> $\in \forall M,E$: \crm
      \Cl \> \> {\bf if} $cc({\cal SS}_i)=M^+(E|M)M^*$ or $M^*(E|M)M^+$ {\bf then} \crm
      \Cl \> \> \>(Let $p_E$ be a point with E)\crm
      \Cl \> \> \> {\bf if} $p_i\neq p_E$ {\bf then} $des_i \leftarrow p_E$ //Possibly the color-configuration becomes $(M|E)$. \crm
      \Cl \> \> {\bf else}// $\#_E({\cal SS}_i)\geq 2$ or $cc({\cal SS}_i)=(M|E)$\crm
      \Cl \> \> \> {\bf if} $l_i=M$ {\bf then} $l_i\leftarrow E$ //$(M|E)(M|E)^*(M|E)\rightarrow EE^*E$ or $(M|E)\rightarrow E$ \crm
      \Cl \> $\in \forall E$: \crm
      \Cl \> \> {\bf if} $cc({\cal SS}_i)=E$ {\bf then} do nothing //Gather\crm
      \Cl \> \> {\bf else if} $cc({\cal SS}_i)=EE$ {\bf then} $l_i\leftarrow S$ //$\rightarrow (S|E)(S|E)$ \crm
      \Cl \> \> {\bf else if} $cc({\cal SS}_i)=EEE$ {\bf then}\crm
      \Cl \> \> \> {\bf if} $p_i = p_n$ {\bf then}\crm
      \Cl \> \> \> \> $l_i\leftarrow S$ //$\rightarrow (S|E)E(S|E)$ \crm
      \Cl \> \> {\bf else} //$cc({\cal SS}_i)=EE^+E$\crm
      \Cl \> \> \> {\bf if} $p_i\neq p_n$ {\bf then} $des_i\leftarrow (p_n+p_f)/2$ //$EE^+E\rightarrow EEE$\crm
      \Cl \> $\in \forall S,M,E$: \crm
      \Cl \> \> {\bf if} $cc({\cal SS}_i)=M^+(S|M|E)M^*$, $M^*(S|M|E)M^+$, or $(S|M|E)$ {\bf and} $l_i=S$ {\bf then}\crm
      \Cl \> \> \> $l_i=E$ //$M^+(S|M|E)M^*\rightarrow M^+(E|M)M^*$ or $(S|M|E)\rightarrow (M|E)$\crm
      \Cl \> \> {\bf else if} $cc({\cal SS}_i)=(S|M)E(S|M)$ {\bf and} $p_i=p_n$ {\bf and} $l_i=S$ {\bf then} \crm
      \Cl \> \> \> $l_i=M$ //$\rightarrow MEM$ \crm
      \Cl \> {\bf endcase}
    \end{tabbing}
    }
  }
\end{algorithm}

$\LU$-Gather-in-ASYNC(Algorithm~\ref{algo:FLGA}) is an extension of Algorithm~\ref{algo:FLG-withcc} so that it can work in ASYNC, and uses color-cycles similar to that of Algorithm~\ref{algo:SIM}. 
This algorithm use $3$ colors $S,M,$ and $E$ and its color-cycle repeats $\forall S(SS)\rightarrow \forall M\rightarrow \forall E(EE)\rightarrow \forall S(SS)$. Notatins in parentheses indicate that the configuration is limited to two points.
\newpage

In the algorithm, robots gather at the midpoint of some onLDS, or   \emph{Gathering point}, where configuration $\conf$ has a Gathering point $p_G$ if and only if $\conf$ is in $\forall M,E$, $\#_E(\conf)=1$ and $p_G$ has $E$.
Thus, the aim of this algorithm to create Gathering point during color-cycles.

In $\forall S(SS)\rightarrow \forall M$, robots on the two points with $S$ change their colors to $M$ and move to the midpoint. 
Note that robots at the endpoints move to the midpoint only if the endpoints have $S$ and the color of the robots is $S$. 
Gathering point is created during transitions in  color-cycles for the following cases; 
\begin{enumerate}
    \item[(1)] During $\forall S(SS)\rightarrow \forall M$, robots with $S$ look  configuration $\conf$ such that $\#_S({\conf})=1$.
\item[(2)] configuration $\conf$ during $\forall M\rightarrow \forall E$.  %
\item[(3)] After configuration becomes $\forall E$ configuration $\conf$ such that$\#_E(\conf)\geq 3$.
\end{enumerate}

We show how to create a Gathering point from each of (1),(2), and (3).%

\noindent
For (1), a robot $r_i$ with $S$ changes its color to $E$ if $r_i$ looks  the configuration with $\#_S({\conf})=1$. From the configuration it will make a configuration such that $\#_E(\conf)=1$ and $\forall M,E$, and then a Gathering point is created.

\noindent
For (2), %
Gathering point is lost if there are more than one point with $E$. On the other hands, 
Gathering point is confirmed if there is only one point with $E$. Let $t_M$ be a time at which the configuration becomes $\forall M$,  let $t_E(>t_M)$ be the first time at which some robot changes its color to $E$, and let $p_E$ be the location having robots with $E$. If there are not activated robots at points except $p_E$ between $t_M+1$ and $t_E$, robots activated after $t_E$ observe a configuration with $\#_E(\conf)=1$ and $\forall M,E$. Therefore, there are no robots that change their colors to $E$, and the Gathering point $p_E$ is confirmed.

\noindent
In (3), robots with $E$ do not move, and the both endpoints are fixed and robots at points except the endpoints move to the midpoint. Thus the color-configuration becomes $EEE$. The transition of the configuration becomes  $EEE\rightarrow SES\rightarrow MEM$, where the last configuration satisfies $\#_E(\conf)=1$ and in $\forall M,E$. This case also determines a Gathering point.

If the configuration becomes (1),  (2) or (3), Gathering point is made. If the configuration does not become (1), (2) and (3), it becomes $EE$, will change $SS$ and again begins the next color-cycle $\forall S(SS)\rightarrow \forall M\rightarrow \forall E(EE)\rightarrow \forall S(SS)$.
If color-cycles are repeated, the distance of the endpoints is reduced by at least $2\delta$ in one cycle. Therefore the distance will become less than $2\delta$ when the color configuration becomes $\forall S(SS)$. Then, if the configuration is in $\forall M$, robots with $M$ reach the midpoint, and Gathering is achieved.
Transitions between color configurations in Algorithm~\ref{algo:FLGA} are shown in Figure~\ref{fig:FLGA}. In this figure, boxes with numbers and  $\curvearrowleft$ are starting configurations and the number corresponds to that in Lemma~\ref{lemma:ColorOnLDS}. Arrow labelled with $L.n(i)$ means it is proved in Lemma~$n(i)$. "DONE" means Gathering is attained.

The following lemmas show transitions between color-configurations in Algorithm~\ref{algo:FLGA}. 
\begin{lemma}\label{lemma:SSS}
If $cc(\conf(t))=SS^+S$, there is a time $t'>t$ such that $cc(\conf(t'))=SS$.
\end{lemma}

\begin{proof}
  If robots on points except the endpoints of onLDS become active, they move to the nearest endpoint (line~8). While $\#_S(\conf)\geq 3$, robots on the endpoints do nothing (line~9) until the color-configuration becomes $cc(\conf)=SS$. Therefore, there is a time t' at which $cc(\conf(t'))=SS$.
  \qed
\end{proof}

\begin{figure}[ht]
  \centering
     \includegraphics[width=0.9 \textwidth,clip]{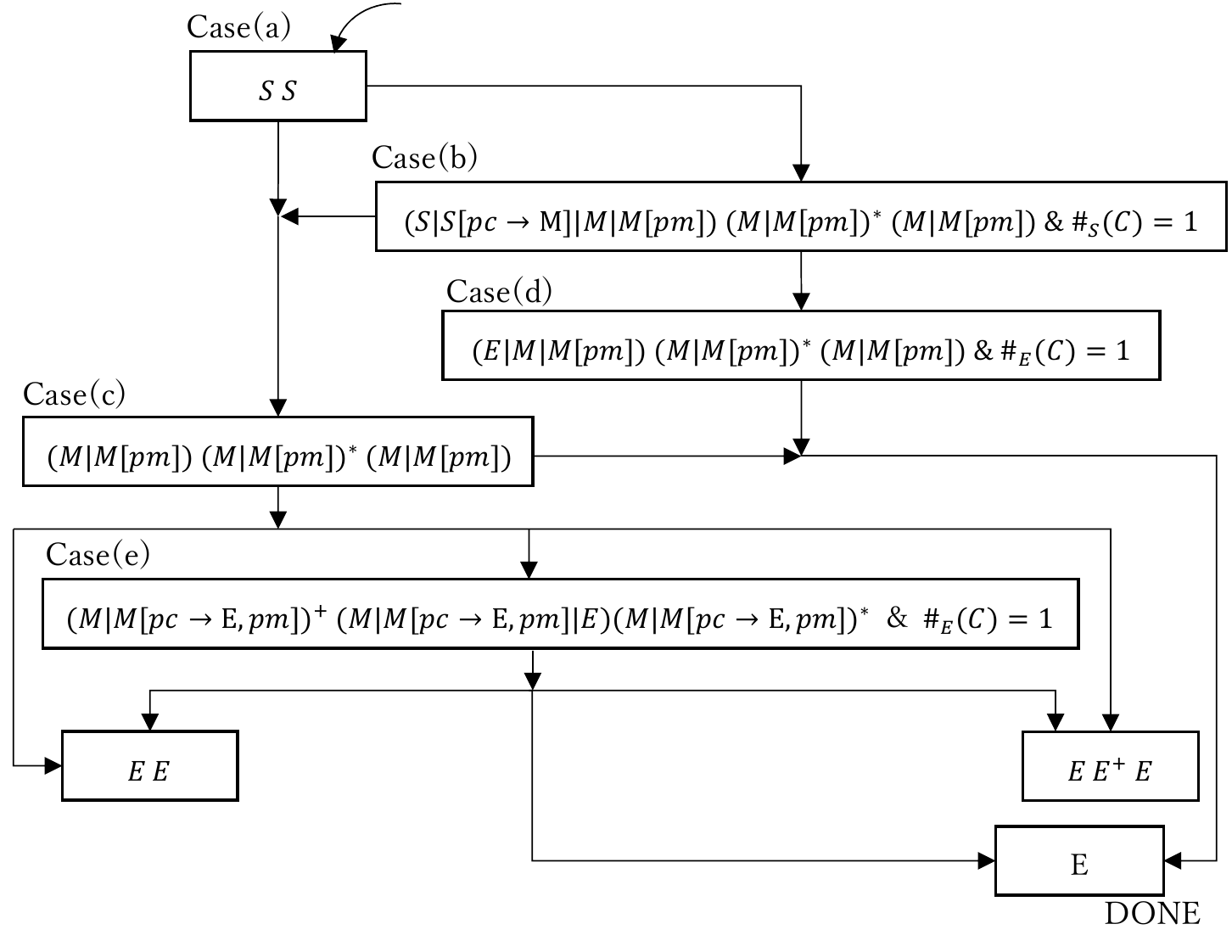}
  \caption{Transition Graph for $\LU$-Gather-in-ASYNC from $cc(\conf)=SS$.}
  \label{fig:FLGALOOP}
\end{figure}

\begin{lemma}\label{lemma:SS}
  If $cc(\conf(t))=SS$, there is a time $t'>t$ such that
  \begin{enumerate}
    \item[(1)]$cc(\conf(t'))=E$, or
    \item[(2)]$cc(\conf(t'))=EE^*E$ and $dis(\conf(t))\leq dis(\conf(t'))$-2$\delta$.
  \end{enumerate}
\end{lemma}

\begin{proof}
Transitions from $cc(\conf(t))=SS$ are depicted in Fig.~\ref{fig:FLGALOOP}.

  {\bf Case(a)} Transition from $cc(\conf(t))=SS$.\\
  We show that
  if $cc(\conf(t))=SS$, there is a time $t_{\forall M}$($>t$) such that $cc(\conf(t_{\forall M}))=(M|M[pm])(M|M[pm])^*(M|M[pm])$ and the destination of robots in $M[pm]$ is the midpoint of onLDS in $\conf(t)$ ({\bf Case~(c)}) or a time $t_{S}$($>t$) such that $cc(\conf(t_S))=(S|S[pc\rightarrow M]|M|M[pm])(M|M[pm])^*(M|M[pm])$, $\#_S(\conf(t_S))=1$ and the destination of robots in $M[pm]$ or $S[pc \rightarrow M]$ is the midpoint of onLDS in $\conf(t)$ ({\bf Case~(b)}).
  
  Let $p_m$ be the the midpoint of onLDS in $\conf(t)$. As long as the both endpoints have $S$, robots with $S$ change their colors to $M$ and will move to $p_m$ as the destination (lines~5-6), and robots with $M$ do nothing when observing a configuration in $\forall S,M$. Thus, if there are no $S$ at the both endpoints at the same time, setting time $t_{\forall M}$ be the last robot(s) with $S$ perform $\comp$,  $cc(\conf(t_{\forall M}))= (M|M[pm])(M|M[pm])^*(M|M[pm])$, because robots with $M$ do nothing when observing a configuration in $\forall S,M$ (line~18). Otherwise, there exists a time when one of the endpoints still has $S$ or $S[pc\rightarrow M]$ but the other endpoint does not have $S$. Setting time $t_S$ be such time, $cc(\conf(t_S))= (S|S[pc\rightarrow M]|M|M[pm])(M|M[pm])^*(M|M[pm])$ and $\#_S(\conf(t_S))=1$ .

  {\bf Case~(b)} Transition from\\ $cc(\conf(t_S))=(S|S[pc\rightarrow M]|M|M[pm])(M|M[pm])^*(M|M[pm])$ ($t_S>t$) and $\#_S(\conf(t_S))=1$\\
  We show that in this case %
  there is a time $t_{\forall M}$($>t_S$) such that $cc(\conf(t_{\forall M}))=(M|M[pm])(M|M[pm])^*(M|M[pm])$({\bf Case~(c)}) or a time $t_{GP}$($>t_S$) such that $cc(\conf({GP}))=(E|M|M[pm])(M|M[pm])^*(M|M[pm])$ and $\#_E(\conf(t_{GP}))=1$\\ ({\bf Case~(d)}). 
  
  In the case that $cc(\conf(t_S))=(S[pc\rightarrow M]|M|M[pm])(M|M[pm])^*(M|M[pm])$, robots in $S[pc\rightarrow M]$ change their colors to $M$ and move to $p_m$. Furthermore, if robots with $M$ finish moving and become active again, they do nothing as long as $S$ exists. Therefore, there is a time $t_{\forall M}$ when $cc(\conf(t_{\forall M}))= (M|M[pm])(M|M[pm])^*(M|M[pm])$, where robots in $M[pm]$ will move to $p_m$. Otherwise, there is a robot with $S$ that observes only one point with $S$, and the robot change its color to $E$ even if the color configuration is $(S|M|E)M^*M$. 
Furthermore, if robots with $M$ finish move and become active again, they do nothing as long as $S$ exists. Therefore, there is a time $t_{GP}$ when $cc(\conf(t_{GP}))= (E|M|M[pm])(M|M[pm])^*(M|M[pm])$ and $\#_E(\conf(t_{GP}))=1$.

  {\bf Case~(c)} Transition from $cc(\conf(t_{\forall M}))=(M|M[pm])(M|M[pm])^*(M|M[pm])$ ($t_{\forall M}>t_S>t$))\\
  We show that if $cc(\conf(t_{\forall M}))=(M|M[pm])(M|M[pm])^*(M|M[pm])$, there is a time $t'$($> t_{\forall M}$) such that $cc(\conf(t'))=E$, a time $t_E$($>t_{\forall M}$) such that $cc(\conf(t_E))=(M|M[pm,pc\rightarrow E])^+(E|M|M[pm,pc\rightarrow E])(M|M[pm,pc\rightarrow E])^*$ and\\ $\#_E(\conf(t_{E}))=1$({\bf Case~(e)}), or a time $t'$($>t_{\forall M}$) such that $cc(\conf(t')=EE^*E$.
  
  If  all robots moved to $p_m$ or are moving toward the same destination $p_m$ since all robots reach $p_m$, the color-configuration becomes $M$ or $(M|E)$. Then, robots with $M$ change their color to $E$ until the color-configuration becomes $E$ by Lemma~\ref{lemma:onePointME}. Thus, there is a time $t'$ such that $cc(\conf(t'))=E$.
  
  If there are robots at only one point that observe $cc({\cal SS})=MM^*M$ and perform \emph{Compute}, the robots change their colors to $E$. Robots on points except the point may observe $cc({\cal SS})=MM^*M$ but do not change their colors. Setting time $t_E$ be such time, $cc(\conf(t_E)=(M|M[pm,pc\rightarrow E])^+(E|M|M[pm,pc\rightarrow E])(M|M[pm,pc\rightarrow E])^*$. Otherwise, there are robots at one or more different points that observe $cc({\cal SS})=MM^*M$ and change their colors to $E$. Then, the configuration satisfies $\#_E(\conf)\geq 2$. Since all robots with $M$ become active by the fairness of ASYNC and change their colors to $E$ until the configuration in  $\forall E$, there is a time $t'$ when the color-configuration becomes $EE^*E$. In this case, all robots have toward the midpoint $p_m$ of LDS in $\conf(t)$, and changed their colors to $E$. Therefore, $dis(\conf(t'))$ is at least $2\delta$ shorter than $dis(\conf(t))$.\\

  {\bf Case~(d)} Transition from $cc(\conf(t_{GP}))=(E|M|M[pm])(M|M[pm])^*(M|M[pm])$\\($t_{GP}>t_S>t$) and $\#_E(\conf(t_{GP}))=1$.
  
  If $cc(\conf(t_{GP}))=(E|M|M[pm])(M|M[pm])^*(M|M[pm])$ and $\#_E(\conf(t_{GP}))=1$, there is a time $t'$($>t_{GP}$) such that $cc(\conf(t'))=E$. Let $p_{GP}$ be a position of the robot with $E$ at $t_{GP}$. Robots not on $p_{GP}$ have colors $M$ and be moving toward $p_m$ or have finished moving. They move toward $p_{GP}$ when observing this configuration. It is possible that robots with $M[pm]$ on $p_{GP}$ move toward $p_m$ and leave $p_{GP}$. However, when they are not on $p_{GP}$, they perform the same action as previously described. Robots with $E$ on $p_{GP}$ do nothing. Thus, the number of positions where a robot has color E do not change, the number of robots with $M[pm]$ decreases. Also, since the number of robots with E increases when $(M|E)$, following from Lemma~\ref{lemma:onePointME}, there is a time $t'$ such that $cc(\conf(t'))=E$.

  {\bf Case~(e)} Transition from $cc(\conf(t_E))=(M|M[pm,pc\rightarrow E])^+\\(E|M|M[pm,pc\rightarrow E])(M|M[pm,pc\rightarrow E])^*$($t_E>t_{\forall M}>t$) and $\#_E(\conf(t_{E}))=1$.\\
  
  We show that for these configurations %
  there is a time $t'$($>t_E$) such that $cc(\conf(t'))=EE^*E$, or a time $t'$($>t_E$) such that $cc(\conf(t'))=E$. 
  
  Let $p_E$ be a position of the robot with E at $t_E$. If a robot not on $p_E$ observes between $t_{\forall M}$ and $t_E$, it changes its color to $E$. When it changes its color, the configuration becomes $\#_E(\conf)\geq 2$. Since all robots with $M$ become active by the fairness of ASYNC and change their colors to $E$ until the configuration in $\forall E$, there is a time $t'$ at which the color-configuration becomes $EE^*E$. In this case, all robots have toward the midpoint $p_m$ of LDS in $\conf(t)$, and changed their colors to $E$. Therefore, $dis(\conf(t'))$ is at least $2\delta$ shorter than $dis(\conf(t))$. Otherwise, robots that observes between $t_{\forall M}$ and $t_E$ are at only $p_E$. Because of the same argument of {\bf Case~(d)}, %
  there is a time $t'$ such that $cc(\conf(t'))=E$.
  \qed
\end{proof}

\begin{lemma}\label{lemma:EE}
  If $cc(\conf(t))=EE$, there is a time $t'$($>t$) such that $cc(\conf(t'))=SS$.
\end{lemma}

\begin{proof}
Robots change their colors to $S$ until the color-configuration becomes $SS$. Since all robots become active after $t$ by the fairness of ASYNC, there is a time $t'$ at which the color-configuration becomes $SS$.
\qed
\end{proof}

\begin{lemma}\label{lemma:EEpE}
  If $cc(\conf(t))=EE^+E$, there is a time $t'(>$t) such that $cc(\conf(t'))=EE_mE$.
\end{lemma}

\begin{proof}
If the color-configuration is $EE^+E$, robots at endpoints do nothing, and robots at points except the endpoints  move to the midpoint $p_m$ of two endpoints in $\conf(t)$. Therefore, moving robot's destination is fixed, and the color-configuration becomes $EE_mE$.
\qed
\end{proof}

\begin{lemma}\label{lemma:EEE}
  If $cc(\conf(t))=EE_mE$, there is a time $t'(>$t) such that $cc(\conf(t'))=SE_mS$.
\end{lemma}

\begin{proof}

If the color-configuration is $EE_mE$, robots at the endpoints change their colors to $S$, and robots at the midpoint do nothing. The color of robots with $E$ at the midpoint do not change until robots at the endpoints change their colors to $S$. Thus, the color-configuration becomes $SE_mS$.
\qed
\end{proof}

\begin{lemma}\label{lemma:SES}
  If $cc(\conf(t))=SE_mS$, there is a time $t'(>$t) such that $cc(\conf(t'))=ME_mM$.
\end{lemma}

\begin{proof}
  If the color-configuration is $SE_mS$, robots at the endpoints change their colors to $M$, and robots at the midpoint do nothing. The color of robots with $E$ at midpoint do not change until robots at the endpoints change their colors to $M$. Thus, the color-configuration becomes $ME_mM$.
  \qed
\end{proof}

\begin{lemma}\label{lemma:MEM}
 If $cc(\conf(t))=ME_mM$, there is a time $t'$($>t$)  such that  $cc(\conf(t'))=E$.
\end{lemma}

\begin{proof}
Let $p_E$ be a position of the robot with $E$ at $t$. Robots not located on $p_E$ move to $p_E$. Since robots on $p_E$ do nothing, there is a time $t'$ at which the color-configuration is $cc(\conf(t'))=(M|E)$.
\qed
\end{proof}

\begin{lemma}\label{lemma:onePointME}
 If $cc(\conf(t))=(M|M[pc\rightarrow E]|E)$ and $\#_M(\conf(t))= 1$, there is a time $t'$($>t$)  such that  $cc(\conf(t'))=E$.
\end{lemma}

\begin{proof}
Robots with $M$ change their colors to $E$ until the color-configuration becomes $E$. Meanwhile, robots with $E$ do nothing. Since all robots become active after $t$ by the fairness of ASYNC, there is a time $t'$ at which the the color-configuration becomes $E$.
\qed
\end{proof}

\begin{lemma}\label{lemma:onePointSM}
 If $cc(\conf(t))=(S|S[pc\rightarrow M or E]|M)$,  $\#_S(\conf(t))= 1$ and $\#_M(\conf(t))= 1$, there is a time $t'$($>t$)  such that  $cc(\conf(t'))=(M|E)$ and $\#_M(\conf(t'))= 1$.
\end{lemma}

\begin{proof}
It is possible that robot with $S$ observes $cc(\conf(t))$ or $cc(\conf)=(S|M|E)$,  $\#_M(\conf(t))= 1$, and $\#_E(\conf(t))= 1$ by the behavior of $S[pc]$. If a robot with $S$ observes $\conf(t)$, the robot changes its color to $E$. If a robot with $S$ observes $cc(\conf)=(S|M|E)$,  $\#_M(\conf(t))= 1$, and $\#_E(\conf(t))= 1$, the robot changes its color to $E$. In both case, robots with $M$ or $E$ do nothing until $\#_S(\conf)$ is $0$. Therefore, there is a time $t'$ such that  $cc(\conf(t'))=(M|E)$ and $\#_M(\conf(t'))= 1$.
\qed
\end{proof}

We show that Algorithm~\ref{algo:FLGA}
can work from the configurations in Lemma~\ref{lemma:ColorOnLDS}. %

\begin{figure}[ht]
  \centering
  \includegraphics[width=0.9 \textwidth,clip]{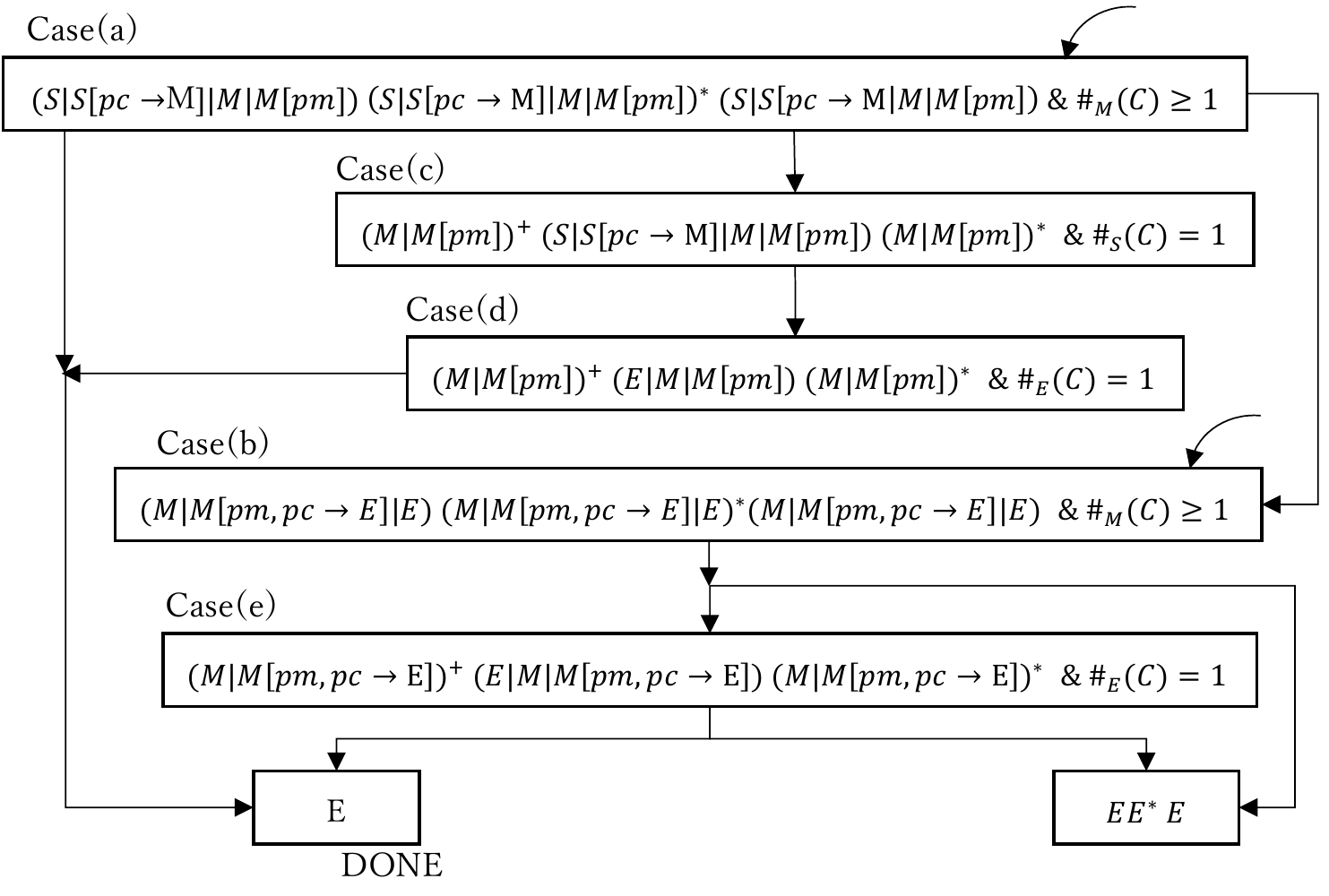}
  \caption{Transition Graph after switching from simulation of ElectOneLDS (Lemma~\ref{lemma:switch})}
  \label{fig:FLGAswitch}
\end{figure}

\begin{lemma}\label{lemma:switch}
  If $cc(\conf(t))$ is one of  (1)-(5) in Lemma~\ref{lemma:ColorOnLDS}, 
  then there is a time $t'$($\geq t$) such that the color-configuration satisfies (a) or (b) in the followings,\\
  \begin{enumerate}
    \item[(a)]$cc(\conf(t'))=EE^*E$,
    \item[(b)]$cc(\conf(t'))=E$.
  \end{enumerate}
\end{lemma}

\begin{proof}
For (1), we can prove it by   Lemmas~\ref{lemma:SSS} and~\ref{lemma:SS}. For (4) and (5), we can prove them by Lemmas~\ref{lemma:onePointSM} and~\ref{lemma:onePointME}, respectively.

For (2) and (3), we can prove it in the followings. Transitions from configurations (2) and (3) are depicted in Fig.~\ref{fig:FLGAswitch}.

\begin{itemize}
    \item {\bf Case(a)} $cc(\conf(t))=$\\ $(S|S[pc\rightarrow M]|M|M[pm])(S|S[pc\rightarrow M]|M|M[pm])^*(S|S[pc\rightarrow M]|M|M[pm])$ and $\#_M(\conf(t))\geq 1$

We show that if {\bf Case (a)} holds, 
there is a time $t_{\forall M}$($\geq t$) such that $cc(\conf(t_{\forall M}))=(M|M[pm,pc\rightarrow E]|E)(M|M[pm,pc\rightarrow E]|E)^*(M|M[pm,pc\rightarrow E]|E)$ ({\bf Case (b)}), a time $t_S$($\geq t$) such that $cc(\conf(t_S))=(M|M[pm])^+\\(S|S[pc\rightarrow M]|M|M[pm])(M|M[pm])^*$ and $\#_S(\conf(t))=1$({\bf Case (c)}), or there is a time $t'$($>t$) such that $cc(\conf(t'))=E$. Note that $cc(\conf(t))$ may be in {\bf Case (b)} or {\bf Case (c)}. 

In the case of $\#_S(\conf(t))\geq 3$, robots with $S$ change their colors to $M$. Thus, $\#_S(\conf(t))$ is decreasing until $\#_S(\conf)\leq 2$. 

If $\#_S(\conf(t))=2$, and the both endpoints of onLDS have $S$, robots with $S$ change their colors to $M$ and move to the midpoint of onLDS. 
Thus, $\#_S(\conf(t))$ decreases to $0$ or $1$.

If $\#_S(\conf(t))=2$, and robots with $S$ is not on the both endpoints of on, robots with $S$ change their color to $M$. Thus, $\#_S(\conf(t))$ decreases to $0$ or $1$. 

If $\#_S(\conf(t))=1$, there is a time $t_S$($>t$) when $cc(\conf(t_S))=(M|M[pm])^+\\(S|S[pc\rightarrow M]|M[pm])(M|M[pm])^*$, or there is a time $t_{op}$($>t$) \\when  $cc(\conf(t_{op}))=(S|S[pc\rightarrow M]|M)$. In the case of  $cc(\conf(t_{op}))=(S|S[pc\rightarrow M]|M)$,  $\conf(t'')$($t''=max\{t_{sc}|\conf(t_{sc})\in \forall S, t_{sc}<t\}$) is symmetric and contractible. Robots with $M[pm]$ that executed ElectOneLDS at $t''$ have same destination that is the center of ${\cal CH}(\conf(t''))$. Therefore, there is a case of $(S|S[pc\rightarrow M]|M)$. We can obtain from Lemma~\ref{lemma:onePointME} and  Lemma~\ref{lemma:onePointSM} that there is a time $t'$ such that $cc(\conf(t'))=E$.

If $\#_S(\conf(t))=0$, robots with $M$ change their colors to $E$. Thus, there is a time $t_{\forall M}$ when $(M|M[pm,pc\rightarrow E]|E)(M|M[pm,pc\rightarrow E]|E)^*(M|M[pm,pc\rightarrow E]|E)$.

  \item {\bf Case (b)}  $cc(\conf(t))=$\\$(M|M[pm,pc\rightarrow E]|E)(M|M[pm,pc\rightarrow E]|E)^*(M|M[pm,pc\rightarrow E]|E)$ and $\#_M(\conf(t))\geq 1$\\
We show that if {\bf Case (b)} holds,
there is a time $t_{E}$($\geq t$) such that $cc(\conf(t_E))=(M|M[pc\rightarrow E,pm])^+(M|M[pm,pc\rightarrow E]|E)(M|M[pm,pc\rightarrow E])$ and\\ $\#_E(\conf(t_E))=1$)({\bf Case (e)}) or a time $t'$($> t$) such that $cc(\conf(t'))=EE^*E$. Note that $cc(\conf(t))$ may be in {\bf Case (e)} or $cc(\conf(t))=EE^*E$. 

In the case of $\#_E(\conf(t))=0$, robots with $M$ change their colors to $E$. If there are robots at only one point that recognize $\#_E(\conf(t))=0$, there is a time $t_E$ when $(M|M[pm,pc\rightarrow E])^+(M|M[pm,pc]|E)(M|M[pm,pc\rightarrow E])^*$ and $\#_E(\conf(t_E)=1$. 

If there are robots at one or more points in different position that recognize $\#_E(\conf(t))=0$, $\#_E(\conf(t))$ is increasing to $2$ or more. Then, robots with $M$ change their colors to $E$ until the configuration in $\forall E$. The same behavior occurs when $\#_E(\conf(t))\geq 2$. Therefore, there is a time $t'$ when $EE^*E$.

  \item {\bf Case (c)} $cc(\conf(t_S))=(M|M[pm])^+(S|S[pc\rightarrow M]|M[pm])(M|M[pm])^*$ and $\#_S(\conf(t))=1$ ($t_S\geq t$)
  
  We show that if {\bf Case (c)} occurs,
there is a time $t_{\forall M}$($>t_S$) such that $cc(\conf(t_{\forall M}))=(M|M[pm,pc\rightarrow E]|E)(M|M[pm,pc\rightarrow E]|E)^*(M|M[pm,pc\rightarrow E]|E)$({\bf Case (b)}) or a time $t_{GP}$($>t_S$) such that $cc(\conf(t_{GP}))=(M|M[pm])^+\\(E|M|M[pm])(M|M[pm])^*$ and $\#_E(\conf(t_{GP}))=1$({\bf Case (d)}).

If all robots with $S$ perform $\look$ before $t_S$, they change their colors to $M$ and may move to a point on onLDS. Furthermore, if robots with $M$ finish moving and become active again, they do nothing as long as $S$ exists. Therefore, $\conf(t_S)$ becomes a configuration in $\forall M$. Then, since robots with $M$ change their colors to $E$,  there is a time $t_{\forall M}$ when $(M|M[pm,pc\rightarrow E]|E)(M|M[pm,pc\rightarrow E]|E)^*(M|M[pm,pc\rightarrow E]|E)$.

there is a time $t_{\forall M}$ when $cc(\conf(t_{\forall M}))= (M|M[pm])(M|M[pm])^*\\(M|M[pm])$. Otherwise, there is a robot with $S$ that does not perform $\look$ before $t_S$, and it changes its color to $E$ when observing $\#_S(\conf)=1$. Furthermore, if robots with $M$ finish moving and become active again, they do nothing as long as $S$ exists. Therefore, there is a time $t_E$ when $cc(\conf(t_{GP}))= (M|M[pm])^+(E|M|M[pm])(M|M[pm])^*$ and $\#_E(\conf(t_{GP}))=1$.

  \item {\bf Case (d)} $cc(\conf(t_{GP}))=(M|M[pm])^+(E|M|M[pm])(M|M[pm])^*$ and\\ $\#_E(\conf(t_{GP}))=1$ ($t_{GP}>t_S$)

We show that if {\bf Case (d)} occurs,
there is a time $t'$($>t_{GP}$) such that $\conf(t')$ is Gathering configuration. Let $p_{GP}$ be a position having robot with $E$ at time $t_{GP}$. Robots not on $p_{GP}$ have color $M$ and may move toward a point on onLDS or be done moving. They move toward $p_{GP}$ when observing this configuration. It is possible that robots with $M[pm]$ on $p_GP$ move toward $p_m$ and leave $p_{GP}$. However, when they are not on $p_{GP}$, they perform the same action as previously described. Robots with $E$ on $p_{GP}$ do nothing. Thus, the number of positions having robots have color $E$ do not change, and so there is a time $t'$ such that $cc(\conf(t'))=E$ by lemma~\ref{lemma:onePointME}.
  
  \item {\bf Case (e)} $cc(\conf(t_{E}))=(M|M[pc\rightarrow E,pm])^+\\(M|M[pm,pc\rightarrow E]|E)(M|M[pm,pc\rightarrow E])$ and $\#_E(\conf(t_E))=1$ ($t_E>t_{\forall M}$)

We show that if {\bf Case (e)} occurs,
there is a time $t'$($>t_E$) such that $cc(\conf(t'))=EE^*E$, a time $t'$($>t_E$) such that $\conf(t')$ is a Gathering configuration.

Let $p_E$ be the position having  robots with $E$ at time $t_E$. If there is a robot not on $p_E$ that observed between $t_{\forall M}$ and $t_E$, it changes its color to $E$. When it changes the color, the configuration has $\#_E(\conf)\geq 2$. Since all robots with $M$ become active by ASYNC scheduler's fairness and change their colors to $E$ until the configuration in $\forall E$, there is a time $t'$ when the color-configuration becomes $EE^*E$. Otherwise, robots that observed between $t_{\forall M}$ and $t_E$ are located at only $p_E$. Because of the same argument of transition from $(E|M|M[pm])(M|M[pm])^*(M|M[pm])$ and $\#_E(\conf)=1$, there is a time $t'$ when $cc(\conf(t'))=E$.
  \end{itemize}
  \qed
\end{proof}
\color{black}

Lemmas~\ref{lemma:ColorOnLDS}-\ref{lemma:switch} follow the following theorem and we obtain our main result.

\begin{theorem} \label{theorem:Gathering3col}
  $\LU$-Gather-in-ASYNC solves Gathering from onLDS for $\LU$ robots having 3 colors, under non-rigid movement.
\end{theorem}

\begin{theorem} \label{theorem:GatheringASYNC}
  Gathering  can be solved in ASYNC by $\LU$  robots having 3 colors under non-rigid movement and agreement of chirality.
\end{theorem}

\section{Concluding Remarks}
\label{sec:conclusion}
We have shown a Gathering algorithm in non-rigid and ASYNC with $\LU$ of three colors. In order to obtain the algorithm, we have shown a simulating algorithm of any algorithm in unfair SSYNC by $\LU$ of three colors in ASYNC. We have reduced the number of colors used in the simulation to three from five, although the simulated algorithms are ones in unfair SSYNC.

The method by combining the simulation of SSYNC robots by ASYNC ones and algorithms working in SSYNC not only reduces the number of colors used in the resultant algorithm but also simplifies the proof of correctness of it. As is known from an example of ElectOneLDS, it seems to be very complicated to extend ElectOneLDS such that  it can work in ASYNC and prove the correctness. However, about correctness of the combined algorithm, it is enough to prove the correctness of the simulation working in ASYNC because the correctness of ElectOneLDS working in SSYNC has been obtained. 

One of the interesting open questions is the number of colors to solve Gathering in ASYNC, although two colors are enough to solve Rendezvous in ASYNC\cite{HDT}, we conjecture that three colors are necessary to solve Gathering in ASYNC.

\paragraph*{Acknowledgment}
This research was partly supported by  JSPS KAKENHI No. 20H04140, 20KK0232, 20K11685, 21K11748, and by Japan Science and Technology Agency (JST) SICORP Grant\#JPMJSC1806.
\bibliographystyle{plain}
\bibliography{referenceorg}

\end{document}